%% file: main.tex
\documentclass{llncs}

\pdfoutput=1

\usepackage[utf8]{inputenc}
\usepackage[T1]{fontenc}
\usepackage{amsmath,mathtools,amssymb}
\usepackage{todonotes}

\usepackage{lmodern}
\usepackage{microtype}
\usepackage{xpunctuate}
\usepackage[all,british]{foreign} 
\redefnotforeign[ie]{i.e\xperiodafter} 
\redefnotforeign[eg]{e.g\xperiodafter}

\usepackage{multirow}
\usepackage{booktabs}
\usepackage{longtable}
\usepackage{pdflscape}
\usepackage{geometry}
\usepackage{colortbl}
\usepackage[round-mode=places,round-pad=false,round-precision=2,group-separator={\,},tight-spacing=true,output-decimal-marker={.},detect-weight=true]{siunitx} 

\usepackage{wrapfig}

\usepackage{amsmath, amssymb, amsfonts} 
\usepackage[full,small]{complexity}

\spnewtheorem{observation}[theorem]{Observation}{\bfseries}{\itshape}

\usepackage{lineno}
\renewcommand{\qed}{\hfill $\square$} 

\PassOptionsToPackage{dvipsnames,usenames,table}{xcolor}
\usepackage{tikz}
\usetikzlibrary{calc}
\usetikzlibrary{matrix,fit,patterns,decorations.pathreplacing}
\usetikzlibrary{arrows,shapes}
\usepackage{tikz-layers}

\usepackage[longend,vlined,plain]{algorithm2e}

\usepackage{marginnote}

\newcommand*\circled[4]{\tikz[baseline=(char.base)]{
             \node[circle,fill=#3,draw=#4,draw,inner sep=1pt,minimum size=1.2em] (char) {\color{#2} #1};}}

\usepackage{environ}
\newcommand{\repeattheorem}[1]{%
  \begingroup
  \renewcommand{\thetheorem}{\ref{#1}}%
  \expandafter\expandafter\expandafter\theorem
  \csname reptheorem@#1\endcsname
  \endtheorem
  \endgroup
}
\NewEnviron{reptheorem}[1]{%
  \global\expandafter\xdef\csname reptheorem@#1\endcsname{%
    \unexpanded\expandafter{\BODY}%
  }%
  \expandafter\theorem\BODY\unskip\label{#1}\endtheorem
}

\newcommand{\repeatlemma}[1]{%
  \begingroup
  \renewcommand{\thelemma}{\ref{#1}}%
  \expandafter\expandafter\expandafter\lemma
  \csname replemma@#1\endcsname
  \endtheorem
  \endgroup
}
\NewEnviron{replemma}[1]{%
  \global\expandafter\xdef\csname replemma@#1\endcsname{%
    \unexpanded\expandafter{\BODY}%
  }%
  \expandafter\lemma\BODY\unskip\label{#1}\endtheorem
}

\newenvironment{tightcenter}
 {\parskip=0pt\par\nopagebreak\centering}
 {\par\noindent\ignorespacesafterend}

\usepackage{ctable}
\newlength{\RoundedBoxWidth}
\newsavebox{\GrayRoundedBox}
\newenvironment{GrayBox}[1]%
   {\setlength{\RoundedBoxWidth}{\textwidth-4.5ex}
    \def\boxheading{#1}
    \begin{lrbox}{\GrayRoundedBox}
       \begin{minipage}{\RoundedBoxWidth}%
   }{%
       \end{minipage}
    \end{lrbox}%
    \begin{tightcenter}%
    \begin{tikzpicture}%
       \node(Text)[draw=black!20,fill=white,rounded corners,%
             inner sep=2ex,text width=\RoundedBoxWidth]%
             {\usebox{\GrayRoundedBox}};
        \coordinate(x) at (current bounding box.north west);
        \node [draw=white,rectangle,inner sep=3pt,anchor=north west,fill=white]
        at ($(x)+(6pt,.75em)$) {\boxheading};
    \end{tikzpicture}%
    \end{tightcenter}\vspace{0pt}%
    \ignorespacesafterend
}

\newenvironment{GrayBoxSlim}[1]%
   {\setlength{\RoundedBoxWidth}{.7\textwidth}
    \def\boxheading{#1}
    \begin{lrbox}{\GrayRoundedBox}
       \begin{minipage}{\RoundedBoxWidth}%
   }{%
       \end{minipage}
    \end{lrbox}%
    \begin{tightcenter}%
    \begin{tikzpicture}%
       \node(Text)[draw=black!20,fill=white,rounded corners,%
             inner sep=2ex,text width=\RoundedBoxWidth]%
             {\usebox{\GrayRoundedBox}};
        \coordinate(x) at (current bounding box.north west);
        \node [draw=white,rectangle,inner sep=3pt,anchor=north west,fill=white]
        at ($(x)+(6pt,.75em)$) {\boxheading};
    \end{tikzpicture}%
    \end{tightcenter}\vspace{0pt}%
    \ignorespacesafterend
}

\title{Efficient Trace Frequency Queries in Sparse Graphs}
\author{%
Christine Awofeso\inst{1}\orcidID{0009-0000-3550-1727} \and
Pål Grønås Drange\inst{2}\orcidID{0000-0001-7228-6640} \and
Patrick Greaves\inst{1}\orcidID{0009-0007-0752-0526} \and
Oded Lachish\inst{1}\orcidID{0000-0001-5406-8121}
Felix Reidl\inst{1}\orcidID{0000-0002-2354-3003}
}
\authorrunning{C. Awofeso et al.}
\institute{Birkbeck, University of London, UK
\email{\{cawofe01|pgreav01\}@student.bbk.ac.uk, \{o.lachish|f.reidl\}@bbk.ac.uk} \and
University of Bergen
\email{Pal.Drange@uib.no}}

\def\any{\mathord{\color{black!33}\bullet}}%
\def\scol_#1{ \operatorname{scol}_{#1} }
\def\wcol_#1{ \operatorname{wcol}_{#1} }
\def\adm_#1{ \operatorname{adm}_{#1} }

\renewcommand{\G}{\mathbb G}
\newcommand{\X}{\mathbb X}

\newcommand{\network}[1]{\texttt{#1}}

\newcommand{\defeq}{\coloneqq} 

\newcommand{\omitted}{\ensuremath{\star}}

\newcommand{\tr}{\operatorname{tr}}
\newcommand{\trcount}{\operatorname{tr}^{\#}}
\newcommand{\DSR}{\mathsf R}
\newcommand{\DSTr}{\mathsf{Tr}}
\newcommand{\DSL}{\mathsf{L}}

\newbox\mytempbox
\tikzstyle{embeds} = [->, >=open triangle 45]

\definecolor{amazing}{RGB}{254,67,101}
\definecolor{cardinal}{HTML}{BB333C}
\definecolor{niceblack}{HTML}{020300}

\tikzset{%
    vertex/.style={%
      circle,fill=niceblack!15,minimum size=18pt,inner sep=0pt
    },%
    red edge/.style={%
      draw,thick,-,cardinal!80
    },%
    black edge/.style={%
      draw,line width=.8pt,-,niceblack
    },%
    gray edge/.style={%
      draw,line width=.8pt,-,niceblack!15
    },%
    box/.style={%
        fill,cardinal,inner sep=5pt,rounded corners=15pt
    }%
}

\SetKwInput{KwInput}{Input}
\SetKwInput{KwOutput}{Output}
\SetKw{Continue}{continue}

\SetCommentSty{mycommfont}

\def\commentmark#1{\circled{#1}{gray}{white}{gray}}
\def\comment#1#2{\textcolor{gray}{\commentmark{#1} #2}}

\def\Nesetril{Ne\v{s}et\v{r}il\xspace}

\def\numnetworks{217\xspace}

\begin{document}

\maketitle
\begin{abstract}
    Understanding how a vertex relates to a set of vertices is a fundamental task in graph analysis. Given a graph~$G$ and a vertex
    set~$X \subseteq V(G)$, consider the collection of subsets of the
    form~$N(u) \cap X$ where~$u$ ranges over all vertices outside~$X$.
    These intersections, which we call the \emph{traces} of~$X$, capture all
    ways vertices in~$G$ connect to~$X$, and in this paper we consider
    the problem of listing these traces efficiently, and the related problem of recording the multiplicity (\emph{frequency}) of each trace.

    For a given query set~$X$, both problems have obvious algorithms with running time $O(|N(X)| \cdot |X|)$ and conditional lower bounds suggest that, on general graphs, one cannot expect better.  However, in certain sparse graph classes, more efficient algorithms are possible: Drange \etal (IPEC 2023) used a data structure that answers trace queries in $d$-degenerate graphs with linear initialisation time and query time that only depends on the query set~$X$ and~$d$. However, the query time is exponential in~$|X|$, which makes this approach impractical.

    By using a stronger parameter than degeneracy, namely the strong $2$-colouring number~$s_2$, we construct a data structure
    in $O(d \cdot \|G\|)$ time, which answers subsequent trace frequency queries in time $O\big((d^2 + s_2^{d+2})|X|\big)$, where $\|G\|$ is the number of edges of $G$, $s_2$ is the strong $2$-colouring number
    and~$d$ the degeneracy of a suitable ordering of~$G$.

    We demonstrate that this data structure is indeed practical and that it beats the simple, obvious alternative in almost all tested settings, using a collection of \numnetworks{} real-world networks with up to 1.1M edges. As part of this effort, we demonstrate that computing an
    ordering with a small strong $2$-colouring number is feasible with a simple heuristic.
\end{abstract}


\section{Introduction}\label{sec:Intro}

Imagine a database with products and customers, modelled as a bipartite graph
where an edge indicates that a customer bought a product. We would like to understand the buying behaviour of a target segment~$X$ of customers better.
A natural operation here is to group the members of~$X$ according to their purchases, specifically,  each subset~$X' \subseteq X$ is of interest if there exists a product~$y$ such that~$y$ was bought exactly by~$X'$ among customers in~$X$.
This allows us to investigate what differentiates~$X'$ from~$X \setminus X'$ and, given \emph{how many} products were bought exactly by~$X'$, determine the relevance of this subgroup. In general, we call the subsets~$X' \subseteq X$ for which there exists a vertex~$y \not \in X$ with~$N(y) \cap X = X'$ the \emph{traces} of~$X$.

\begin{GrayBox}{Example}\small
    \begin{minipage}{.6\textwidth}
    In this example, the \emph{traces} of~$X$ are the sets
    \vspace*{-8pt}
    \begin{itemize}
        \item $\{b,c\}$ (neighbours of~$2$ and~$4$),
        \item $\{b\}$ (neighbour of~$3$),
        \item $\{c\}$ (neighbour of $5$),
        \item $\{c,d\}$ (neighbours of~$6$),
        \item $\{d\}$ (neighbour of $7$), and
        \item $\emptyset$ (vertices~$1$ and~$8$ have no neighbours in~$X$).
    \end{itemize}
    \end{minipage}\hspace{10pt}%
    \begin{minipage}{.38\textwidth}
    \includegraphics[width=\textwidth]{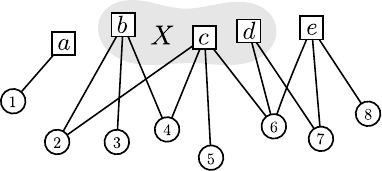}
    \end{minipage}\\[2pt]
    The \emph{frequencies} of each trace are the number of vertices that generate it, for example, the frequency of
    $\{b,c\}$ is $|\{2,4\}| = 2$.
\end{GrayBox}
\vspace*{-10pt}

\noindent
The problem we investigate here is as follows: given a graph~$G$ (not necessarily bipartite) and a query set~$X \subseteq V(G)$, list the traces of~$X$ and count their frequencies. Additionally, we will look at the related problems of only listing the traces, as well as the simpler problem of counting the number of neighbours of~$X$, \ie the size of~$|N(X)|$. Note that any algorithm or data structure that counts trace frequencies implicitly solves this latter problem as well.

A natural approach to count trace frequencies is to compute for each~$y \in N(X)$ the trace $N(y) \cap X$ and store it in a suitable data structure. This results in a running time of at least~$\Omega(|X| \cdot |N(X)|)$, where the factor $|N(X)|$ is potentially very large.
Moreover, since the output of this query can have size~$2^{|X|}$, it is clear that we need to restrict ourselves to graphs without such adverse structures to arrive at a reasonable running time.
Maybe more surprisingly, we show via a simple reduction that already the task of counting the number of neighbours of~$X$ cannot be solved in time~$O(2^{o(|X|)})$, unless the Hitting Set Conjecture fails (see Section~\ref{sec:lower-bounds} for details).

\paragraph{The sparse setting.}
We therefore focus on sparse graph classes that are relevant in practice.
One important property of many such classes is precisely that the number of traces of any given set~$X$ is bounded by a function of~$|X|$. In $d$-degenerate graphs\footnote{Recall that a graph is $d$-degenerate if every subgraph contains a vertex of degree at most~$d$.}, it is well-known that the number of traces is bounded by~$O(|X|^d)$ (we include a proof in Section~\ref{sec:trace-bounds} for completeness). Since degeneracy of real-world networks is typically small---the average degeneracy in our dataset is around $21$---this class is a good starting point.
In a previous paper~\cite{drange2023computing} with a very different problem setting, a subset of the authors
together with Muzi designed a data structure that, after an initialisation
time of $O(d2^d |G|)$ in a $d$-degenerate graph~$G$, answers trace
frequency queries in time~$O(|X|2^{|X|} + d|X|^2)$. The factor~$2^{|X|}$ in
this specific data structure is unavoidable, as it uses inclusion-exclusion
over~$|X|$ to answer the query. There were two key problems with this approach we look to overcome. First, while in the context of~\cite{drange2023computing} the query
sets were known to be small, we here are interested in scenarios where~$X$ is potentially very large. Second, the data structure in~\cite{drange2023computing} uses $\Omega(2^d |G|)$ space, which is impractical for larger networks with even moderate degeneracy.

Motivated by this, we focus on sparseness parameters that provide tighter bounds on the number of traces.
Many graph classes---like planar graphs, graphs of bounded degree, or graphs excluding a (topological) minor---admit a linear bound, \ie every set~$|X|$ has at most~$O(|X|)$ traces.
A unified `explanation' of this fact can be found in the theory of \emph{bounded expansion classes} (which include the above classes) as described by \Nesetril and {Ossonda de Mendez}~\cite{Sparsity}. Specifically, they note that the \emph{strong $2$-colouring number} $\scol_2(G)$ governs the number of traces in a graph~$G$, and that $\scol_2(G)$ is a constant for all graphs from bounded expansion classes~\cite{zhuColouring2009}. We define $\scol_2$ in Section~\ref{sec:preliminaries} and show how it bounds the number of traces in Section~\ref{sec:trace-bounds}.

\paragraph{Our contribution.}
We design a data structure that, given a suitable ordering of the input graph with degeneracy~$d$ and strong $2$-colouring number $s_2$, is initialised with time and space complexity~$O(d \|G\|)$ and then answers
trace frequency queries in time $O((d^2 + s_2^{d+2})|X|)$. Further, the properties of the strong $2$-colouring number allow us to store the necessary information using bit vectors of length~$s_2$, greatly reducing the memory consumption in practice.

To demonstrate that the data structure is indeed practical, we conducted several computational experiments on a set of \numnetworks{} real-world networks from various domains, with up to 1.1M edges. The results are summarized in Section~\ref{sec:experiments} and a more detailed breakdown can be found in the Appendix.

In the first set of experiments, we show that finding suitable orderings with low strong $2$-colouring numbers is achievable with a very simple heuristic, even though the problem is \NP-hard~\cite{chen_strong_2006}. We use a related measure that provides upper and lower bounds for the strong $2$-colouring number and which can be computed exactly in practice~\cite{twoAdmissibility25} to argue that the heuristic works well on the \numnetworks{} tested networks.

In the second set, we compare the performance of our data structure in answering trace, trace frequency, and neighbourhood counting queries against the obvious baseline algorithms for different query sizes. We deemed that a comparison against the data structure by Drange \etal~\cite{drange2023computing} was not necessary as the results in that paper already show that the exponential dependency on the query size is prohibitive in practice.

The results show that for trace listing and trace frequency queries, our data structure outperform the naive algorithms in almost all scenarios. For neighbourhood counting, the naive algorithm is superior.

\section{Preliminaries}\label{sec:preliminaries}

\marginnote{$\any$}
We will often use the placeholder $\any$ for variables whose value is irrelevant in the given context.

\marginnote{$[\any]$, $\X$, $<_\X$, $\max_\X$, $\min_\X$}%
For an integer $k$, we use $[k]$ as a short-hand for the set $\{0, 1, 2, \ldots, k-1\}$. We use blackboard bold letters like~$\X$ to denote totally ordered sets, that is, some underlying set~$X$ associated with a
total order~$<_\X$. We further, for~$Y \subseteq X$, use the notations~$\max_\X Y$ to mean the maximum vertex in~$Y$ under~$<_\X$ and the similarly defined~$\min_\X Y$. If the context allows it, we will sometimes drop this subscript, \eg we shorten $\max_\X X_1 <_\X \max_\X X_2$ to~$\max X_1 <_\X \max X_2$ or~$\max_\X \X$ to simply~$\max \X$.


\marginnote{$|G|$, $\|G\|$}
All graphs in this paper will be undirected and simple.
For a graph $G$ we use $V(G)$ and $E(G)$ to refer to its vertex- and edge-set. We use the shorthands $|G| \defeq |V(G)|$ and $\|G\| \defeq |E(G)|$.

\paragraph*{Traces and trace frequency.}

\marginnote{Traces, $\tr_G(X)$, $\tr_G(X\mid Y)$}
The \emph{trace} $\tr_G(X)$ of vertex set~$X$ in graph~$G$ is the set of all
subsets~$X' \subseteq X$ for which there exists a vertex~$y \in G\setminus X$
such that~$N(y) \cap X = X'$. That is,
\[
    \tr_G(X) := \{ X' \subseteq X \mid \exists y \in V(G)\setminus X ~\text{with}~N(y) \cap X = X'\}.
\]
In the following it will be useful to constrain the set $Y \subseteq V(G)$ from which we choose the vertices~$y$. To that end, we use the notation
\[
  \tr_G(X \mid Y) :=  \{ X' \subseteq X \mid \exists y \in Y\setminus X ~\text{with}~N(y) \cap X = X'\}.
\]
A useful observation is that~$\tr_G(X \mid Y)$ and $\tr_G(X \mid Y \cap N(X))$
can only differ by at most one set and that this is exactly the case when
$\emptyset \in \tr_G(X \mid Y)$:

\begin{observation}\label{obs:trace-diff}
    For all~$X, Y \subseteq V(G)$ it holds that
    $\tr_G(X \mid Y)$ contains the same traces as $\tr_G(X \mid Y \cap N(X))$,
    with the exception of the empty set which might be contained in the former
    but never appears in the latter. Therefore,
    $$
    |\tr_G(X \mid Y \cap N(X))|
    \leq
    |\tr_G(X)|
    \leq
    1 + |\tr_G(X \mid Y \cap N(X))|
    .$$
\end{observation}

\marginnote{Trace frequencies, $\trcount_G(X)$, $\trcount_G(X\mid Y)$}
\noindent
An important related quantity is the \emph{trace frequency} $\trcount_G(X)$, which not only records which traces in~$X$ appear in~$G$ but also how often. We understand $\trcount_G(X)$ as a multiset of traces which implicitly records the
frequency of each trace. $\trcount_G(X\mid Y)$ similarly counts the frequencies of traces induced in~$X$ by vertices from~$Y$.

\paragraph*{Degeneracy and strong $2$-colouring number.}

\marginnote{$\G$, ordered graph}
An \emph{ordered graph} is a pair $\G = (G, <)$ where $G$ is a graph and $<$ a
total ordering of $V(G)$. We write $<_\G$ to denote the ordering for a given
ordered graph and extend this notation to the derived relations $\leq_\G$,
$>_\G$, $\geq_\G$. For simplicity we will call $\G$ an \emph{ordering of} $G$
and we write~$\pi(G)$ to denote the set of all possible orderings.

\marginnote{$N^-, \Delta^-$}
We use the same notations for graphs and ordered graphs, additionally we write
$N_\G^-(u) \defeq \{ v \in N(u) \mid v <_\G u \}$ for the \emph{left neighbourhood} of a vertex $u \in \G$. We write $N_\G^-[u] \defeq N_\G^-(u) \cup \{u\}$ for the closed left neighbourhood which we extend to vertex sets~$X$ via $N_\G^-[X] \defeq \bigcup_{u \in X} N_\G^-[u] \cup X$. We write~$\Delta^-(\G) \defeq \max \{ |N_\G^-(u)| \mid u \in \G \}$ to denote the maximum left-degree of an ordered graph.


\marginnote{degeneracy}
A graph~$G$ is \emph{$d$-degenerate} if there exists an ordering $\G$ such that
$\Delta^-(\G) \leq d$.
We say an ordering~$\G$ of a graph~$G$ is~$d$-degenerate if   $\Delta^-(\G) \leq d$, and just \emph{degenerate} if $\Delta^-(\G)$ is minimum. The degeneracy ordering of a graph can be computed in time $O(n + m)$ and~$O(dn)$ for $d$-degenerate graphs~\cite{matulaDegeneracy1983}.

\marginnote{$S^2(\any)$, $S^2[\any]$, strongly $2$-reachable}
The \emph{strongly $2$-reachable set} for a vertex~$u$ in an ordered graph~$\G$ is the set~$S_\G^2(u) := N^-(N^+(u) \cup u) \cap V_{\leq u}$. That is, $S_\G^2(u)$ contains all vertices that are smaller than~$u$ and
are neighbours of~$u$, or can be reached via a path of length~$2$ where the mid-point is larger than~$u$.

\smallskip
\includegraphics[width=.8\textwidth]{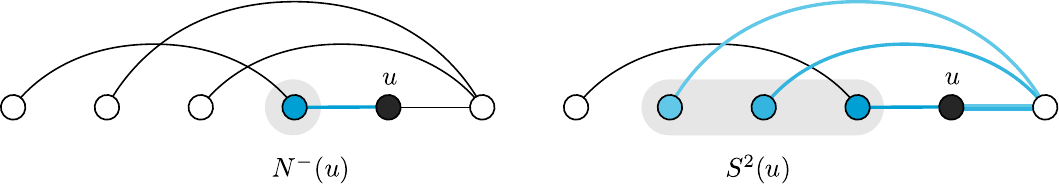}

\noindent
We will often need the set~$S^2_\G(u) \cup \{ u \}$, therefore we write~$S^2_\G[u]$ for this `closed' version of the set. In both cases, we usually drop the subscript~$\G$ if the ordering is clear from the context.

\marginnote{Strong $2$-colouring number}
The \emph{strong $2$-colouring number} $\scol_2(\G)$ is then defined
as the maximum size of~$S^2$, \ie $\scol_2(\G) := \max_{u \in \G} |S_\G^2(u)|$. For an unordered graph~$G$, the strong $2$-colouring number is defined as the minimum over all possible ordering of~$G$, \ie
$\scol_2(G) := \min_{\G \in \pi(G)} \scol_2(\G)$.

The following observation describes the important property of~$S^2$ that lies at the heart of our trace data structure:
\begin{observation}\label{obs:S2}
    Let~$\G$ be an ordered graph, $u \in \G$ and~$x \in N^-(u)$.
    Then~$N^-(u) \cap V_{\leq x} \subseteq S^2[x]$.
\end{observation}
\begin{proof}
  Clearly, every vertex $y \in V_{\leq x}$ is no larger than~$x$ in~$\G$ and hence
  every $y \in N^-(u) \cap V_{\leq x}$ is no larger than~$x$.
  Since~$xu \in E(\G)$, the path $xuy$ exists. Thus~$y \in S^2[x]$.
\end{proof}

\noindent
One important special case is when~$x = \max_\G N^-(u)$ is the largest among the left neighbours of~$u$, in that case we have that~$N^-(u) \subseteq S^2[x]$.

\section{Bounding the number of traces by $s_2$ and~$d$}\label{sec:trace-bounds}

The first important observation is that in a degenerate ordering~$\G$, the number of traces induced by vertices that are smaller than the query set's maximum cannot be large---simply because the number of such vertices is bounded:

\begin{lemma}\label{lemma:trace-bounds-left}
    Let~$\G$ be a $d$-degenerate ordering of a graph~$G$ and let~$X \subseteq V(G)$. Let further~$z := \max_\G X$ be the largest vertex in~$X$ under~$<_\G$.
    Then $|\tr_{G}(X \mid V_{\leq z})| \leq d \cdot |X| + 1$.
\end{lemma}
\begin{proof}
    Simply note that~$|N(X) \cap V_{\leq z}| \leq \Delta^-(\G) \cdot |X| = d \cdot |X|$.
    By Observation~\ref{obs:trace-diff},
    \[
      |\tr_G(X \mid V_{\leq z})| \leq |\tr_G(X \mid V_{\leq z} \cap N(X))| + 1
      \leq  d \cdot |X| + 1. \tag*{\qed}
    \]
\end{proof}

\noindent
The interesting bounds therefore concern the traces induced by vertices that are \emph{larger} than the maximum of the query set. The bounds based on the degeneracy and the strong $2$-colouring number illustrate already why the latter might be more beneficial for our purposes. Note that the degeneracy-bound is well-known, we include it here for completeness. Bounds using~$s_2$ are known indirectly (\cite{kernelStructuralBndExp13,kernelDomsetBndExp2016,neighbourhoodComplexity2019} prove bounds for measures related to~$s_2$), but as far as we are aware the below bound is more precise:

\begin{lemma}\label{lemma:trace-bounds}
    Let~$\G$ be an ordering of a graph~$G$ and let~$X \subseteq V(G)$. Let
    $d$ be the degeneracy and $s_2$ the strong $2$-colouring number of~$\G$.
    Then the following bounds hold:
    \begin{enumerate}
        \item $|\tr_G(X)| \leq |X|^d + d \cdot |X| + 1$
        \item $|\tr_G(X)| \leq |X| s_2^d + d \cdot |X| + 1$
    \end{enumerate}
\end{lemma}
\begin{proof}
    Let~$z := \max_{\G} X$ be the largest vertex in the query set~$X$.
    We prove the bounds by bounding~$|\tr(X \mid V_{> z})|$ by $|X|^d$
    and $s_2^d |X|$, respectively.  The claimed bounds then follow by
    seeing that $\tr(X) = \tr(X \mid V_{\leq z}) \cup \tr(X \mid V_{> z})$
    and applying Lemma~\ref{lemma:trace-bounds-left} to the first set.

    For the first bound, simply note that for every vertex~$y \in V_{> z}$ it holds that~$N(y) \cap X = N^-(y) \cap X$. Since~$|N^-(y)| \leq d$, it follows that these vertices can induce at most
    \[
        \sum_{i = 0}^{d} {|X| \choose i} \leq  |X|^{d} + 1
    \]
    traces and hence~$|\tr(X \mid V_{> z})| \leq |X|^d + 1$. Note that this bound counts the empty set which is already counted by the bound from  Lemma~\ref{lemma:trace-bounds-left}. We can therefore
    combine the two bounds and subtract one to arrive at the claimed bound.

    For the second bound we will set up a recursive inequality over the size of the trace set and the largest vertex in the trace. To that end, let
    $$t^u_i := |\{ S \in \tr(G\mid V_{> z}) \mid |S| = i \text{ and } \max_\G S = u\}$$
    denote the number of traces in $\tr(G \mid V_{> z})$ of size exactly~$i$ and with~$u$ as their largest vertex. As observed above, the maximum size is~$d$ and thus~$t^u_i = 0$ for~$i > d$. As the induction basis we use~$t^u_1 = 1$ for all~$u \in X$.

    Note that if vertices~$v$ and~$u$ with $v \leq_\G u$ appear in a trace together, then~$v \in S^2(u)$. Therefore we can bound~$t^u_i$ by
    $
        t^u_i \leq \sum_{v \in S^2(u)} t^v_{i-1} \leq \sum_{v \in S^2(u)} \sum_{w \in S^2(v)} t^w_{i-2} \leq \cdots \leq s_2^{i-1}
    $.
    Therefore,
    $
        |\tr_G(X \mid V_{>z})| \leq 1 + \sum_{u \in X} \sum_{i=1}^{d} t^u_i
        \leq 1 + |X| \sum_{i=1}^d s_2^{i-1} \leq 1 + |X| s_2^d
    $.
    This bound again counts the empty set and we can subtract one when combining it with the bound from Lemma~\ref{lemma:trace-bounds-left}.
\end{proof}

\noindent
Note that it is trivial to construct a $d$-degenerate graph~$G$ where, for some set~$X$,~$|\tr_G(X)| \geq |X|^d$. As such, for the task of \emph{outputting} traces a data structure based on $d$-degeneracy cannot have a query complexity lower than~$O(|X|^d)$. Whether it is possible to \emph{count} the number of traces faster than that is an interesting open question, in Section~\ref{sec:lower-bounds} we show some conditional lower bounds for data structure that only count the number of neighbours.

\section{The trace data structure}\label{sec:trace-data-structure}

In this section we fix an ordering~$\G$ of our input graph~$G$ and let
$d$ be the degeneracy\footnote{Note that $d$ is the degeneracy of the ordering and might be larger than the graph's actual degeneracy. A theoretical worst-case here is that~$d = s_2$. In practice, we find that $d$ is much smaller than~$s_2$.} and~$s_2$ the strong $2$-colouring number of~$\G$. While computing an optimal ordering~$\G$ with~$s_2 = \scol_2(\G)$ is \NP-hard~\cite{hardnessGeneralizedColoring2023}, as we demonstrate below in Section~\ref{subsec:two-colouring-practice}, in practice it is easy to find orderings with low~$\scol_2$ values.

 In the following, we will assume that~$d, s_2 \geq 1$ as~$d = 0$ and~$s_2 = 0$ are only possible in
edgeless graphs which we will exclude from this discussion.
We first describe the underlying data structure~$\DSR$ in theory and discuss its implementation further below. On a high level, $\DSR$ is a two-level associative array with vertices of~$G$ as the first key, vertex subsets as the second key, and an integer as the stored value. We initialise~$\DSR$ as described in Algorithm~\ref{alg:init}.

\begin{algorithm}[tb]
\begin{GrayBoxSlim}
\DontPrintSemicolon
\KwInput{An ordered graph~$\G$}
Initialize $\DSR$ as an empty associative array\;

\For{$u \in \mathbb G$}{
  \For{$x \in N^-(u)$}{
    $K \leftarrow N^-(u) \cap V_{\leq x}$  \tcp*[r]{Takes time $O(d)$}
    $\DSR[x][K] \leftarrow \DSR[x][K] + 1$ \tcp*[r]{Non-existing keys are treated as zero}
  }
}
\Return $\DSR$\;
\end{GrayBoxSlim}
\caption{Initialisation of the data structure $\DSR$}
\label{alg:init}
\end{algorithm}

\begin{observation}\label{obs:init-content}
    After initialisation, for every vertex~$u$ the data structure~$\DSR[u]$
    contains exactly the sets $\tr_G( S^2[u] \mid V_{>u})$ as keys and their
    multiplicities in $\trcount_G(S^2[u] \mid V_{>u})$ as values,
    \eg $\DSR[u][X]$ tells us how many vertices in~$V_{>u}$ have~$X$ as their
    trace in~$S^2[u]$.
\end{observation}
\begin{proof}
     By Observation~\ref{obs:S2} for every vertex $y \in V_{>u}$ we have
     that~$N^-(y) \cap V_{\leq u}$ is a subset of~$S^2[u]$, thus the keys
     of~$\DSR[u]$ are exactly~$\tr_G(S^2[u] \mid V_{>u})$. The claim about the multiplicity is easy to verify from the code. \qed
\end{proof}

\begin{lemma}\label{lemma:R-complexity}
    The initialisation of~$\DSR$ takes time~$O(d \|G\|) = O(n d^2)$ and the complete data structure uses the same space;~$O(d \|G\|) = O(nd^2)$. Moreover, for every vertex~$u \in V(G)$, $\DSR[u]$ stores at most~$3s_2^{d+1}$ keys.
\end{lemma}
\begin{proof}
    The running time is straightforward. For every vertex~$u$ in the outer loop, we store~$|N^-(u)| \leq d$ prefix sets as keys in $\DSR$. The number stored in the dataset are all bounded above by~$|G|$ and therefore take unit space in the RAM model. Accordingly, we use a total of~$O(d |G|)$ space.

    \noindent
    We are left to bound the number of keys in~$\DSR[u]$.

    By Observation~\ref{obs:init-content}, for every vertex~$u$ the data structure~$\DSR[u]$ has $\tr_G(S^2[u] \mid V_{> u})$ as keys. In the proof of Lemma~\ref{lemma:trace-bounds}, we showed that for any set~$X$ and vertex~$z$ the quantity $\tr_G(X | V_{> z})$ is bounded by $|X| s_2^d + 1$. Accordingly, the number of keys in~$\DSR[u]$ is at most
    \[
        \tr_G(S^2[u] \mid V_{> u}) \leq |S^2[u]| s_2^d + 1
        \leq (s_2 + 1) s_2^d + 1 \leq 3s_2^{d+1},
    \]
    where we used that~$d,s_2 \geq 1$ in the last step. \qed
\end{proof}

\noindent
We now describe the algorithm which uses the data structure
$\DSR$ to answer a trace query:

\begin{algorithm}[bt]
\begin{GrayBoxSlim}
\DontPrintSemicolon
    \KwInput{A vertex set~$X \subseteq G$ with ordering~$x_1, \ldots, x_\ell$ in~$\G$}
    \KwOutput{An associative array~$\DSTr$ containing~$\trcount(X)$}
    \DontPrintSemicolon
    Initialize $\DSTr$ as an empty associative array storing integers\;
    $\DSTr[\emptyset] \leftarrow |V(G)\setminus X|$\;
    \comment{1}{Collect `right' traces}\;
    \For{$i \in [\ell]$}{
        \For{$A \in \DSR[x_i]$}{
            $S \leftarrow A \cap X$\;
            $\DSTr[S] \leftarrow \DSTr[S] + \DSR[x_i][A]$\;
            $S' \leftarrow S \setminus \max_{\G} S$\;
            $\DSTr[S'] \leftarrow \DSTr[S'] - \DSR[x_i][A]$\;
        }
    }
    \comment{2}{Collect and correct `left' traces}\;
    Initialize $\DSL$ as an empty associative array storing vertex lists\;
    \For{$i \in [l]$}{
        \For{$u \in N^-[x_i]$}{
            \If{$u \not\in \DSL$}{
                $S \leftarrow N^-(u) \cap X$\;
                $\DSTr[S] \leftarrow \DSTr[S] - 1$ \hspace*{1em}\comment{2a}{Remove incorrect trace}\;
                $\DSL[u] \leftarrow S$ \hspace*{4.65em}\comment{2b}{Insert left neighbours of~$u$}\;
            }
            $\DSL[u] \leftarrow \DSL[u] \cup \{x_i\}$
        }
    }
    \For{$u \in \DSL$}{
        \If{$u \not \in X$}{
            $S \leftarrow \DSL[u]$\;
            $\DSTr[S] \leftarrow \DSTr[S] + 1$ \hspace*{2.5em}\comment{2c}{Count correct trace}\;
        }
    }
    \Return $\DSTr$\;
\end{GrayBoxSlim}
\caption{
    Answering trace queries using the data structure $\DSR$.
}
\label{alg:trace}
\end{algorithm}

\begin{lemma}
    Algorithm~\ref{alg:trace} answers the trace frequency query for a vertex
    subset~$X \subseteq V(G)$ in time~$O\big((d^2+s_2^{d+2}) |X|\big)$.
\end{lemma}
\begin{proof}
    We first prove the correctness and then the running time of the algorithm.

    \smallskip
    \noindent \textbf{Correctness}.
    Let in the following~$z := \max_\G X$ be the largest vertex in~$X$.
    Fix~$X' \subseteq X$. We now show that at the end of Algorithm~\ref{alg:trace}, the data structure contains the correct count for~$X'$, that is, $\DSTr[X']$ exactly the multiplicity of~$X'$ in $\trcount(X)$.
    Let~$Y'$ contain all vertices whose trace in~$X$ is~$X'$, including vertices in~$X$.

    \begin{claim}
        After completion of part \commentmark{1}, $\DSTr[X']$
        contains the number of vertices~$y \in Y'$ for which~$N^-(y) \cap X = X'$.
    \end{claim}
    \begin{proof}
        Let~$x_i := \max_\G X'$, then when the outer loop of part \commentmark{1} reaches~$i$, by Observation~\ref{obs:init-content}, $\DSR[x_i]$
        contains exactly~$\tr(S^2[x_i] \mid V_{> x_i})$ as keys. Note that in all
        previous iterations, neither the set~$S$ nor the set~$S'$ can be equal to~$X'$ since in those iterations, $x_i \not \in S$.

        Consider first the case where no vertex in~$V_{> x_i}$ has~$X'$ as their trace in~$x_1,\ldots,x_i$. In that case, neither the set~$S$ nor the set~$S'$ will equal~$X'$ and therefore $\DSTr[X']$ remains unchanged.

        If there exists any vertex $y \in V_{> x_i}$ with~$N(y) \cap \{x_1,\ldots,x_i\} = X'$, then note that~$X' \subseteq S^2[x_i]$.
        Let~$Y_t$ contain all vertices of~$V_{> x_i}$ where~$N^-(y) \cap X = X'$ and $Y_f$ all vertices of~$V_{> x_i}$ where~$N^-(y) \cap X \supset X'$. We now show that~$\DSTr[X']$ equals $|Y_t|$ at the end of part \commentmark{1}.

        By Observation~\ref{obs:init-content}, $\DSR[x_i][A]$ contains the number of vertices in $V_{>x_i}$ whose trace in $S^2[x] $ is~$A$. So every vertex~$y \in Y_y \cup Y_f$ is counted
        in~$\DSTr[X']$ at exactly the iteration where $N^{-}(y) \cap S^2[x_i] = X'$. For a vertex~$y \in Y_f$, let~$x_j$ be the next vertex in
        $N^-(y) \cap X$ in the ordering~$x_1, \ldots, x_\ell$, by construction
        of~$Y_f$ such an index must exist. Consider the iteration of the
        inner loop when~$A = N^-(y) \cap \{x_1,\ldots,x_j\}$. Note that in
        this iteration $S' = X'$ since~$x_i$ is the immediate predecessor of~$x_j$ in~$A \cap X$ under~$\G$. In this iteration, the contribution of
        $y$ to $\DSTr[X']$ is removed. Therefore, after iterating through all
        $x_1, \ldots, x_\ell$, we conclude that $\DSTr[X']$ indeed contains
        the number of vertices vertices~$y \in Y'$ for which~$N^-(y) \cap X = X'$.
    \end{proof}

    \noindent
    Note that at this point, the traces of all vertices in~$V_{> z}$ are correctly recorded in~$\DSTr$. However, the data structure also contains incorrect traces for some vertex in~$V_{\leq z}$. Namely, those vertices~$y$
    who lie somewhere between~$x_1$ and~$x_\ell$ in~$\G$. These vertices
    currently contribute to~$\DSTr[N^-(y) \cap X]$ but their actual trace in~$X$ might be different to~$N^-(y) \cap X \neq N(y)$. This is corrected in part~\commentmark{2}: note that for any such vertex~$y$ with~$N^-(y) \cap X \neq N(y) \cap X$ it follows that~$y \in N^-[X]$. Their contribution to~$\DSTr[N^-(y) \cap X]$ is removed in line~\commentmark{2a}, note that this line is executed exactly once for each vertex in~$N^-[X]$ as the vertex is inserted into~$\DSL$ immediately afterwards. To see that the correct trace is counted for each $u \in N^-[X]$, we first show the following:

    \begin{claim}
        After the first loop in part~\commentmark{2} has finished, $\DSL$ contains exactly~$N^-[X]$ as keys and for all $u \in N^-[X]$,
        we have that $\DSL[u] = N(u) \cap X$.
    \end{claim}
    \begin{proof}
        It is clear that all keys of~$\DSL$ are indeed from $N^-[X]$ as these are the only vertices used for insertions. To see that it is all of $N^-[X]$, simply note that for each~$u \in N^-[X]$ there exists at least one $i \in [l]$ such that $u \in N^-[x_i]$. Once~$i$ takes this value in the loop, $u$ is inserted into~$\DSL$.

        Consider now~$u \in N^-[X]$. When~$u$ is inserted into~$\L$ in line~$\commentmark{2b}$ we have that $\DSL[u] = N^-(u) \cap X$. Then for every right neighbour $x_i \in N^+(u) \cap X$, at loop iteration~$i$, $x_i$ is added to $\DSL[u]$. We conclude that when the loop terminates, $\DSL[u] = (N^-(u) \cap X) \cup (N^+(u) \cap X) = N(u) \cap X$ and the claim holds. \qed
    \end{proof}
    Finally, in the second loop of part~\commentmark{2}, the $\DSL[u]$
    for~$u \in N^-(X)$ (so excluding~$X$) is inserted into $\DSTr$.
    We conclude that after part~\commentmark{2}, $\DSTr$ indeed contains the trace frequencies of~$X$.

    \smallskip
    \noindent \textbf{Running time analysis}.
    As stated, we assume that that the associative arrays $\DSTr$, $\DSR$ are constructed using a hash map and therefore have constant expected query time. The inner data structures~$\DSR[x]$ for~$x \in G$ have a query time proportionally to the size of the key set, which can be achieved using \eg prefix tries. \looseness-1

    Let us begin by analysing the running time of part \commentmark{1}. By Lemma~\ref{lemma:R-complexity}, the number of elements in~$\DSR[x_i]$
    as at most~$3s_2^{d+1}$. Accordingly, both loops taken together perform at most~$|X| 3s_2^{d+1}$ iterations. In each iteration, note that we can
    bound the size of all three sets ($S, S', A$) by~$|A|$ and since
    $A \subseteq S^2[x_i]$ this in turn is bounded by~$O(s_2)$. Therefore all access operations ($\DSTr[S], \DSTr[S'], \DSR[x_i][A]$) are bounded by~$O(s_2)$ and the total running time of part \commentmark{1} is at most~$O( s_2^{d+2}|X|)$.

    Let us now analyse the running of part~\commentmark{2}. The two nested loops run in time~$\sum_{x_i \in X} |N^-[x_i]| \leq d |X|$. For every vertex in~$N^-[X]$, the if-statement inside these loops is executed exaclty once. Note that~$S$ has size at most~$d$, thus constructing~$S$
    and accessing~$\DSTr[S]$ is possible in time~$O(|S|) = O(d)$. The update of~$\DSL[u]$ after the if-statement takes constant time, so in total the running time of the first loop is bounded by~$O(d^2 |X|)$.

    For the second loop it costs~$O(|S|)$ to update every trace~$S$ of vertices
    in~$N^-(X)$. That is, the running time cost is proportional to
    $\sum_{u \in N^-(X)} N(u) \cap X$ which is equal to the total number of edges between $X$ and~$N^-(X)$. This in turn is bounded by~$d (|X| + |N^-(X)|) = O(d^2 |X|)$ and the claimed running time follows. \looseness-1 \qed
\end{proof}

\begin{GrayBox}{\small Variants}\small
The algorithm can be made slightly more efficient to count only the number of neighbours of a query set, the code is listed in the Appendix.
For graphs that contain additional information in the form of vertex labels/colours, we can compute a data structure independently for each label to answer traces queries relative to those labels, \eg for a set~$X' \subseteq X$ we want to know how many vertices of a certain label have the trace~$X'$ in~$X$. For practical purposes it is trivial to modify the data structure to support vertex labels/colours directly.
\end{GrayBox}
\vspace*{-.6em}

\begin{GrayBox}{\small Implementation notes}\small
While the theoretical analysis already suggests that the data structure is quite space-efficient, Observation~\ref{obs:S2} holds the key for great practical gains: Because the data structure $\DSR[x]$, $x \in G$, only stores keys for subset of~$S^2[x]$ and the size of~$S^2$ is usually small (in the tens or hundreds for most networks, see Section~\ref{sec:experiments}), we can store all keys in~$S^2[x]$ as bit vectors of length~$|S^2[x]|$ over the ground set~$S^2[x]$. This is orders of magnitudes smaller than storing vertex sets that contain vertex ids, where each vertex id realistically has a size of at least 32 bit, and has second-order effects on cache efficiency. This also means that the set intersection in loop \commentmark{1}
of Algorithm~\ref{alg:trace} can be computed as a bitwise AND between two bit vectors, which again is an order of magnitude faster than doing the same with two hash sets. \looseness-1
\end{GrayBox}
\vspace*{-1.6em}

\section{A simple conditional lower bound}~\label{sec:lower-bounds}

\noindent
As discussed above, answering trace queries in $d$-degnerate graphs is lower bounded by~$\Omega(|X|^d)$ simply by the potential output size,
though there is no immediate reason why \eg \emph{counting} the number of traces should not be possible in less time. Here we provide some weaker, conditional lower bounds already for counting the number of neighbours of a query set. To that end, recall the following conjecture formulated by Abboud, V. Williams, and Wang~\cite{abboudRadius2016}:

\begin{definition}[Hitting Set Conjecture (HSC)]
There is no $\epsilon > 0$ such that for all~$c \geq 1$ there exists an algorithm that given two lists~$\mathcal A, \mathcal B$ of~$n$ subsets of a universe~$U$ of size~$c \log n$, can decide in~$O(n^{2-\epsilon})$ time whether there exists a set in~$\mathcal A$ that intersects every set in~$\mathcal B$.
\end{definition}

\noindent
The problem described in the definition is called the \emph{Hitting Set Existence} (HSE) problem. The following conditional lower bound for neighbourhood-counting data structures follows from a straightforward adaption of a construction by Abboud \etal~\cite{abboudRadius2016}:

\begin{replemma}{repLowerBound}[\omitted]
    For every data structure that can be initialised in time $f_1(d)|G|$on a $d$-degenerate graph~$G$ and then count the size of~$N[X]$ for a given query set~$X \subseteq V(G)$ in time~$f_2(d+|X|)$, it holds
    that one of~$f_1(x), f_2(x)$ cannot be in~$O(2^{o(x)})$ unless the HSC fails.
\end{replemma}

\noindent
The same construction shows that for every data structure that can be initialised in time~$f_1(|G|)|G|$ on a graph~$G$ and count the neighbourhood sizes for query sets~$X \subseteq V(G)$ in time~$f_2(|X|)$, it is not possible that $f_1(x) \in o(x)$ and $f_2(x) \in O(2^{o(x)})$ unless the HSC fails.

\section{Experiments}\label{sec:experiments}

We implemented\footnote{
    Code available at \url{https://github.com/microgravitas/exp-trace-sampling}
} the data structure and comparison algorithms in Rust (v1.88.0) and ran the experiments on an Apple M2 @3.49 GHz with 8GB RAM
on a single core. Our test dataset\footnote{
    Networks available at \url{https://github.com/microgravitas/network-corpus}
} contains \numnetworks{} networks ranging from eight to 1.1M edges from various domains like biology, infrastructure, sociology, and communication.

\subsection{Computing strong $2$-colouring orders in practice}\label{subsec:two-colouring-practice}

\begin{figure}[t]
    \includegraphics[width=.96\textwidth]{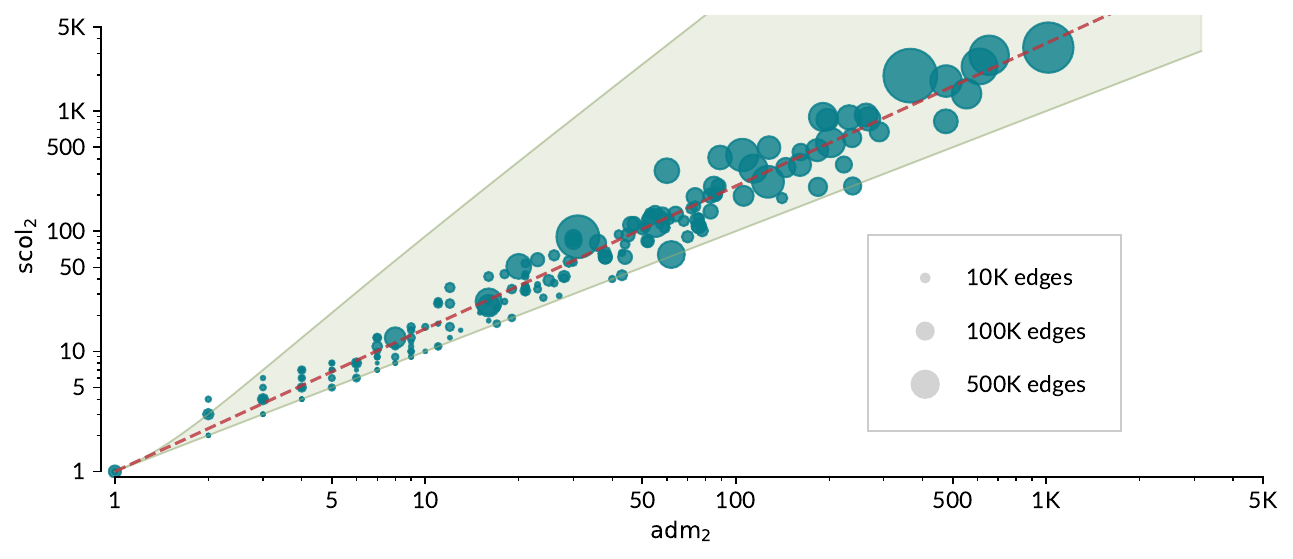}
    \vspace*{-1em}
    \caption{
        Strong $2$-colouring numbers of degeneracy/degree orderings for all \numnetworks{} networks in the corpus, compared to their exact $2$-admissibility. The shaded region indicates the lower and upper bounds based on $2$-admissibility orderings, proving that  degeneracy/degree orders have very low strong $2$-colouring numbers in practice. The red dashed line is the fitted polynomial
        $\adm_2(G)^{1.188}$.
    }
    \label{fig:scol-vs-adm}
\end{figure}

An immediate practical concern of using decompositions like the strong $2$-colouring order is the decomposition has to be readily computable. As mentioned above, the problem of finding an optimal $\scol_2$-ordering is \NP-hard, however, two simple heuristics work very well in practice, namely by computing either a degeneracy ordering (which is possible in linear time~\cite{matulaDegeneracy1983} and very fast in practice) or by simply ordering the vertices by ascending degree. In our experiments, we computed both and chose the ordering with the better $\scol_2$-value.

Since we do not have an optimal $\scol_2$-values as a baseline to compare these heuristics, we
instead use an algorithm recently published by a subset of the authors~\cite{twoAdmissibility25} that can compute the \emph{$2$-admissibility} $\adm_2$ exactly, even on very large networks.

\smallskip
\begin{GrayBox}{A quick explanation of admissibility}\small
To measure the $2$-admissibility $\adm_2(\G)$ of an ordering~$\G$, we ask for every vertex~$u$ how many paths of length~$2$ can be packed together that start in~$u$ and terminate in a vertex left of~$u$, that is, in $S^2(u)$. Here `packing' means that the only common vertices of these paths is~$x$ itself. The two-admissibility $\adm_2(G)$ of a graph~$G$ is then the minimum over all orderings of~$G$.
\end{GrayBox}
\vspace*{-1.4em}

\noindent
For any graph~$G$ it holds that $\adm_2(G) \leq \scol_2(G) \leq \adm_2(G) \cdot (\adm_2(G)-1)$~\cite{dvorakDomset2013}, therefore the $2$-admissibility provides a context to judge the strong $2$-colouring number of heuristically computed orderings.

\begin{figure}[t]
    \includegraphics[width=.96\textwidth]{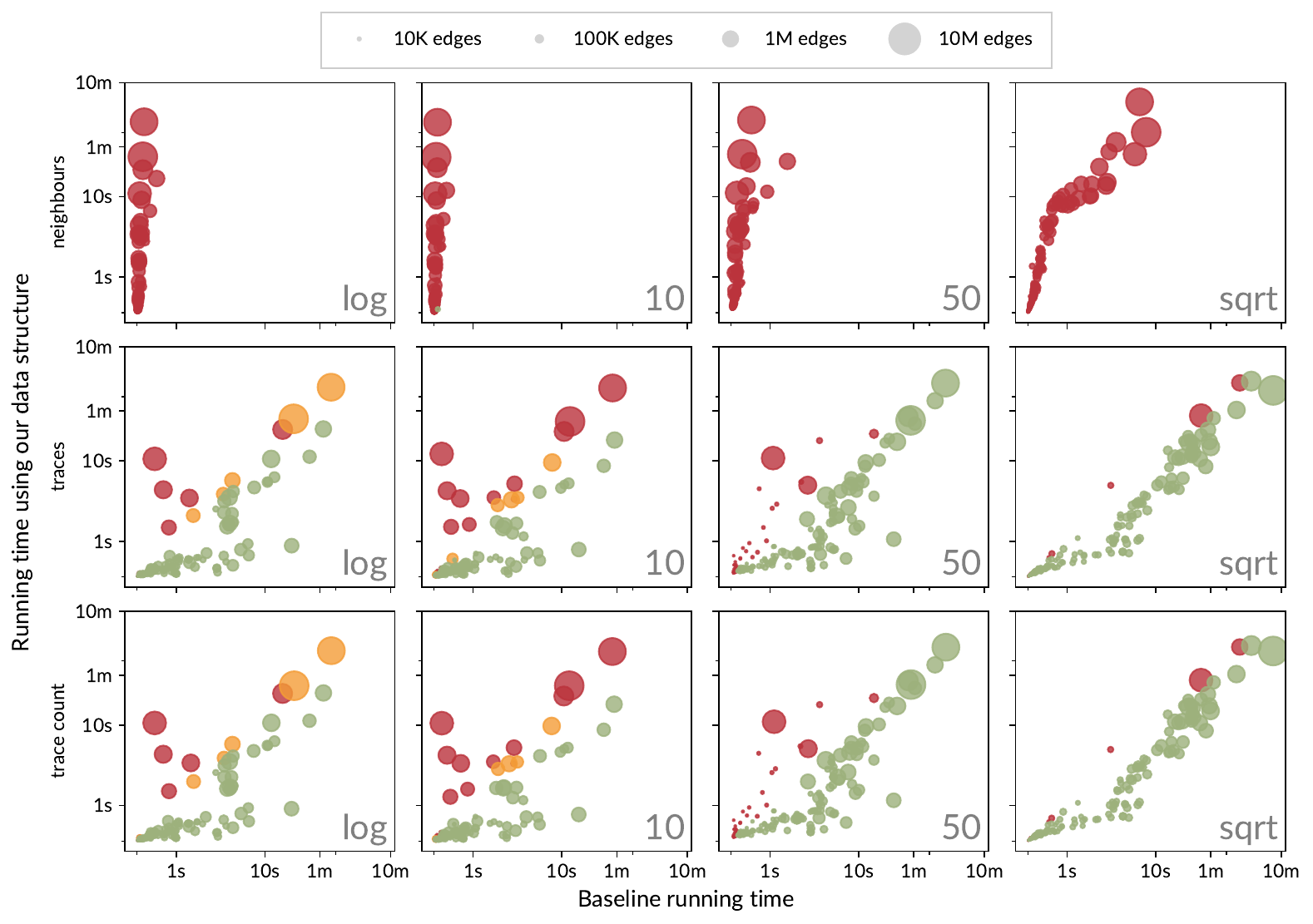}
    \vspace*{-.5em}
    \caption{
        Comparison of baseline algorithms to algorithms using our data structure for neighbourhood-counting (top), trace enumeration (middle), and trace frequency (bottom).
        The running time for the baseline algorithm is on the x-axis,
        the running time of our algorithm on the y-axis.
        The time was measured for 1000 query sets chosen uniformly at random of size $\log_2 |G|$, 10, 50, and $\sqrt{|G|}$.
        Red indicates that the baseline algorithm was faster, orange that it was faster only if including the setup time of our data structure, and green that our algorithm was faster including the setup time. Marker sizes correspond to~$\|G\|$.
    }
    \label{fig:query-times}
\end{figure}

For the \numnetworks{} networks in our corpus, we found that the degeneracy/degree heuristic computed orderings with $\scol_2$-values much closer to the lower bound given by~$\adm_2$ (\cf Figure~\ref{fig:scol-vs-adm}). Testing different fitting functions (linear, polynomial, exponential) between $s_2$ and the $2$-admissibility $a_2$, we found the best fit ($R^2 = 0.917$) for $s_2 \sim a_2^{1.188}$ in our dataset. The average $\scol_2$-value computed was $155.8$, with a 75\% of the networks having a value below 110 and 50\% below 25. The largest value we found was 3379 in a network with 426K vertices and 8M edges (\network{dogster\_friendships}). $s_2$ has almost no correlation with the number of vertices ($r \approx 0.24$) and correlates positively with the number of edges ($r \approx 0.73$). There is only a weak positive correlation with the average degree ($r \approx 0.39$), so density is not the only factor at play here, indicating that the networks in the dataset have various structural differences.

\subsection{Comparing queries against baseline algorithms}

We compare our data structure on three tasks to generic baseline algorithms:
counting the number of neighbours for a query set~$X$, computing $\tr
(X)$, and computing $\trcount(X)$. We implemented the baseline algorithms
using sensible defaults, like using state-of-the art hash maps, but otherwise
stuck to the `obvious' algorithm one would implement to solve the particular
problem.

Since all three tasks are sensitive to the query size, we tried four different
query set sizes: $\log_2 |G|$ vertices, ten vertices, fifty vertices, and
$\sqrt{|G|}$ vertices. On about half of the networks we have that $\log_2
|G| \leq 10$, and the overall maximum is 21. Larger query sets favour our
algorithm, therefore we did not test sets beyond $\sqrt{|G|}$.

For each network~$G$, we generated 1000 query sets chosen uniformly at random from~$V(G)$ and ran our algorithm and the baseline algorithm on the same set. The setup cost for our data structure is included in the running time, but it only has to be computed once to answer all 1000 queries.

As Figure~\ref{fig:query-times} illustrates, our algorithm is preferable in almost all cases for the trace counting and listing task, in particular for larger query sizes (50, $\sqrt{|G|}$). On smaller queries ($\log_2 |G|$, 10), it performs worse on very large graphs which either have a large $s_2$, \eg \network{livemocha} with 2.1M edges and $s_2 = 2347$, or have such low edge density that most queries will be answered very quickly by the basic algorithm, \eg \network{mag\_geology\_coauthor} with 4.4M edges and an average degree of $3.1$.

 Regarding the neighbourhood counting, the basic algorithm seems preferable in all circumstances.

\section{Conclusion} 

We designed and tested a data structure to answer trace listing and trace
frequency queries based on orderings with small $2$-colouring number. We
experimentally confirmed that the data structure is indeed practicable, as
already suggested by its favourable theoretical properties. In our
experiments, the data structure performed better than the `obvious' basic
algorithms in almost all settings, and we presented a very simple heuristic
to compute the necessary ordering that performs well on all tested networks.

We pose as an open question whether it is possible to improve these types of queries on $d$-degenerate graphs, specifically whether 
there exists a data structure that can be initialised in linear time and answers trace frequency queries in time~$O(f(d) |H|^d)$.

\bibliographystyle{plain}
\bibliography{biblio}

\newcommand{\abbr}[1]{\textcolor{gray}{#1}}
\appendix
\section*{Appendix}

\begin{algorithm}[h]
\begin{GrayBoxSlim}
\DontPrintSemicolon
    \KwInput{A vertex set~$X \subseteq G$ with ordering~$x_1, \ldots, x_\ell$ in~$\G$}
    \KwOutput{An associative array~$\DSTr$ containing~$\trcount(X)$}
    \DontPrintSemicolon
    $c = 0$\;
    $\DSTr[\emptyset] \leftarrow |V(G)\setminus X|$\;
    \comment{1}{Count `right' neighbours}\;
    \For{$i \in [\ell]$}{
        \For{$A \in \DSR[x_i]$}{
            \If{$A \cap X = \{x_i\}$}{
                $c \leftarrow c + \DSR[x_i][A]$\;
            }
        }
    }
    \comment{2}{Count `left' neighbours}\;
    $c \leftarrow c + |N^-[X]|$\;
    \comment{3}{Correct double counts}\;
    \For{$u \in X$}{
        \If{$N^-(u) \cap X \neq \emptyset$}{
            $c \leftarrow c - 1$\;
        }
    }
    \Return $c$\;
\end{GrayBoxSlim}
\caption{
    Answering neighbourhood-counting queries using the data structure $\DSR$.
}
\label{alg:neighbourhood-counting}
\end{algorithm}

\repeatlemma{repLowerBound}
\begin{proof}
    Given an instance of~$(\mathcal A, \mathcal B)$ of the HSE problem over universe~$U$ with~$|U| = c\log n$ for some constant~$c$, construct the triparite graph~$G$ with vertex sets~$\mathcal A,\mathcal B,U$ and
    edges~$Xu$ for~$X \in \mathcal A \cup \mathcal B$ and~$u \in U$ whenever~$u \in X$. The HSE instance is a YES-instance iff
    there exists~$X \in A$ such that~$N_G(X) \cap N_G(Y) \neq \emptyset$ for all~$Y \in B$.

    Assume there exists a datastructure that on a $d$-degenerate graph~$G$ can be initialised in time~$O(f_1(d) |G|)$ and then count the number of neighbours of a given query set~$X \subseteq V(G)$ in time~$O(f_2(d+|X|))$.
    Note that the graph~$G$ constructed above is~$|U|$-degenerate, as witnessed by any ordering that first contain all of~$U$ and then~$A \cup B$. We use the datastructure to solve the HSE problem as follows:
    first, we prepare the data structure to answer queries on~$G' = G[U \cup B]$ in time~$O(f_1(|U|) |G'|)$. Then, for each set~$X \in \mathcal A$, we query the datastructure to obtain the neighbour count
    for~$X$ in~$G'$, if the answer is equal to~$|B|$ we return YES. If we do not find any member of~$\mathcal A$ where this happens, we return NO. The total running time of this algorithm is
    \begin{align*}
        O\big(f_1(|U|) \cdot |G'|\big) + O\Big(\sum_{X \in \mathcal A} f_2(|U| + |X|)\Big)
        &= O\big((f_1(c\log n))  + f_2(2c\log n)) \cdot n\big)
    \end{align*}
    Assume that~$f_1(x), f_2(x) \in O(2^{o(x})$, then the above running time becomes~$O\big(2^{o(\log n)} \cdot n\big) = O(n^{1+o(1)})$, contradicting the HSC.
\end{proof}

\newpage
\section*{$\scol_2$ experiments}

\def\sdegree{s_2^\text{degree}}
\def\sdegen{s_2^\text{degen}}
In the following we will refer to the heuristical $\scol_2$-values of a network as $\sdegree$ and $\sdegen$ and $s_2 = \min\{\sdegree, \sdegen\}$ for the combined value. We will to refer to the (exact) $2$-admissibility value
derived from the experiments in~\cite{twoAdmissibility25} as~$a_2$. Recall that~$\adm_2(G) \leq \scol_2(G) \leq \adm_2(G)(\adm_2(G)-1)$ for any graph~$G$. For the computed values this means that~$a_2 \leq s_2$, but it is possible that~$s_2 > a_2(a_2-1)$.

We could additionally use a $2$-admissibility ordering as an approximation algorithm to enforce the upper bound, but this comes at a higher computational cost than the other two heuristics. As we see in the results below, it is also not necessary since in all but one small example, the $s_2$ value lies far below the upper bound $a_2(a_2-1)$ and much closer to~$a_2$: When we combine both heuristics, 70\% of networks satisfy~$s_2 \leq 2a_2$,
and the largest factor is $s_2 \approx 5.4 a_2$.

When comparing the two heuristics, $\sdegree = \sdegen$ on 45 networks, $\sdegree > \sdegen$ on 64 networks, and $\sdegree < \sdegen$ on 105 networks. The degree ordering heuristic seems to work better on networks with larger average degree and on bipartite networks ($\sdegree = \sdegen$ on 45 bipartite networks, $\sdegree > \sdegen$ on 18, $\sdegree < \sdegen$ on 10), but we could otherwise find no clear pattern.

\medskip
\noindent
The two values differ on average by about 26\%, but there are outliers. The \network{chicago}
network ($n = 1467, m = 1298$) is a very sparse road transportation network and has with $\sdegree = 11$ and $\sdegen = 1$ the largest relative difference between the two values. The \network{tv\_tropes} network ($n = 152\,093, m = 3\,232\,134$), which represents which artistic work employs which common story trope as per \url{tvtropes.org}, has with~$\sdegree = 2916$ and $\sdegen = 3715$ the largest absolute difference.

\medskip
\noindent
In the table below, the $\scol_2$-columns labelled with \emph{degen} and \emph{degree} correspond to the values~$\sdegen$ and~$\sdegree$, the column \emph{min} contains~$s_2$. The better of the two values $\sdegen$, $\sdegree$ is highlighted in gray. The column ${\adm_2}^{\underline 2}$ refers to the quantity $a_2 \cdot (a_2-1)$, the upper bound to~$\scol_2$.


\definecolor{Gray}{gray}{0.9}
\newrobustcmd{\best}{{\cellcolor[gray]{0.9}}}
\footnotesize
\include{tables/adm-scol}

\newpage
\section*{Query time experiments}

The following is a detailed breakdown of our algorithms compared to naive baseline algorithms on the three problems. The numbers indicate on how many of the \numnetworks networks each algorithm performed better. A sum less than \numnetworks indicates either that some networks timed out after 10 minutes, or in the case of the constant-sized queries that some networks were smaller than the query size (one network has less than 10 vertices, 12 networks have less than 50 vertices). The number in parenthesis indicate the number of borderline networks, here the baseline performs worse when the initialisation time is ignored, indicating that for a larger number of queries our algorithm would at some point improve over the baseline.

\begin{center}
\begin{tabular}{ll@{\hskip 50pt}l@{\hskip 15pt}l@{\hskip 15pt}l@{\hskip 15pt}l@{}}
          &            & \multicolumn{4}{c}{Query size} \\  \cmidrule{3-6}
  Problem & Algorithm  & $\log_2 |G|$ & $10$ & $50$ & $\sqrt{|G|}$ \\ \toprule
  \multirow{2}{3cm}{Neighbourhood counting} &
  Baseline
              &  215 (0) & 213 (0) & 203 (0) & 215 (0) \\
  & Ours
              &  0 & 1 & 0 & 0 \\ \midrule
  \multirow{2}{3cm}{Trace listing} &
  Baseline
              & 18 (5) & 33 (7) & 63 (0) & 15 (0) \\
  & Ours
              & 197 & 181 & 140 & 199 \\ \midrule
  \multirow{2}{3cm}{Trace frequencies} &
  Baseline
              & 18 (6) & 28 (5) & 61 (0) & 14 (0) \\
  & Ours
              & 197 & 186 & 142 & 200 \\ \bottomrule
\end{tabular}
\end{center}

\noindent
The following pages contain the complete experimental data, all reported times are in seconds. We abbreviated some network names for the sake of space.

\savegeometry{default}
\newgeometry{margin=1in, landscape}
\begin{landscape}
\definecolor{Gray}{gray}{0.9}
\newcolumntype{R}{>{\columncolor{Gray}}S}
\newcolumntype{C}{>{\columncolor{Gray}}S}
\footnotesize
\include{tables/queries}
\end{landscape}
\loadgeometry{default}

\end{document}

%% file: tables/adm-scol.tex
\begin{longtable}{@{}l@{}S[table-format=7.0]S[table-format=7.0]S[table-format=4.2]S[table-format=5.0]S[table-format=3.0]S[table-format=4.0]S[table-format=4.0]S[table-format=4.0]S[table-format=4.0]S[table-format=7.0]@{}}
 &  &  &  &  &  & \multicolumn{3}{c}{{$\scol_2$}} &  &  \\ \cmidrule{7-9}
{Network} & {$m$} & {$n$} & {$\bar d$} & {$\Delta$} & {deg} & {\emph{degen}} & {\emph{degree}} & {\emph{min}} & {$\adm_2$} & {${\adm_2}^{\underline 2}$} \\ 
\midrule
\endhead
\bottomrule
\endfoot
AS-oregon-1 & 23409 & 11174 & 4.19 & 2389 & 17 & 50 & \best 42 & 42 & 28 & 756 \\
AS-oregon-2 & 32730 & 11461 & 5.71 & 2432 & 31 & 87 & \best 83 & 83 & 52 & 2652 \\
\abbr{BG}-\abbr{AC}-\abbr{Lumin.} & 2312 & 1840 & 2.51 & 376 & 6 & \best 11 & \best 11 & 11 & 8 & 56 \\
\abbr{BG}-\abbr{AC}-Ms & 321887 & 40495 & 15.90 & 2217 & 58 & \best 473 & 612 & 473 & 183 & 33306 \\
\abbr{BG}-\abbr{AC}-Rna & 42815 & 13765 & 6.22 & 3572 & 54 & \best 124 & \best 124 & 124 & 75 & 5550 \\
\abbr{BG}-\abbr{AC}-Western & 64046 & 21028 & 6.09 & 535 & 17 & \best 139 & 215 & 139 & 64 & 4032 \\
\abbr{BG}-All & 1316843 & 75550 & 34.86 & 3620 & 134 & \best 1779 & 1979 & 1779 & 476 & 226100 \\
\abbr{BG}-\abbr{A.}-Thaliana-Columbia & 47916 & 10417 & 9.20 & 1341 & 26 & \best 129 & 171 & 129 & 53 & 2756 \\
\abbr{BG}-Biochemical-Activity & 17746 & 8620 & 4.12 & 427 & 11 & \best 56 & 81 & 56 & 29 & 812 \\
\abbr{BG}-Bos-Taurus & 424 & 454 & 1.87 & 27 & 3 & \best 5 & \best 5 & 5 & 4 & 12 \\
\abbr{BG}-\abbr{C.}-Elegans & 23646 & 6394 & 7.40 & 522 & 64 & \best 90 & 100 & 90 & 70 & 4830 \\
\abbr{BG}-\abbr{C.}-Albicans-Sc5314 & 1609 & 1121 & 2.87 & 427 & 9 & \best 10 & \best 10 & 10 & 9 & 72 \\
\abbr{BG}-Canis-Familiaris & 125 & 143 & 1.75 & 90 & 2 & \best 2 & \best 2 & 2 & 2 & 2 \\
\abbr{BG}-Chemicals & 28093 & 33266 & 1.69 & 413 & 1 & \best 1 & \best 1 & 1 & 1 & 0 \\
\abbr{BG}-Co-Crystal-Structure & 2021 & 2291 & 1.76 & 92 & 5 & \best 5 & \best 5 & 5 & 5 & 20 \\
\abbr{BG}-Co-Fractionation & 56354 & 11017 & 10.23 & 187 & 83 & \best 146 & 190 & 146 & 83 & 6806 \\
\abbr{BG}-Co-Localization & 4452 & 3543 & 2.51 & 63 & 6 & 17 & \best 13 & 13 & 9 & 72 \\
\abbr{BG}-Co-Purification & 5970 & 4326 & 2.76 & 1972 & 12 & \best 16 & \best 16 & 16 & 12 & 132 \\
\abbr{BG}-Cricetulus-Griseus & 57 & 69 & 1.65 & 30 & 1 & \best 1 & \best 1 & 1 & 1 & 0 \\
\abbr{BG}-Danio-Rerio & 266 & 261 & 2.04 & 61 & 3 & \best 3 & 4 & 3 & 3 & 6 \\
\abbr{BG}-\abbr{D.}-Discoideum-Ax4 & 20 & 27 & 1.48 & 4 & 1 & \best 1 & \best 1 & 1 & 1 & 0 \\
\abbr{BG}-Dosage-Growth-Defect & 2193 & 1447 & 3.03 & 213 & 5 & 20 & \best 15 & 15 & 9 & 72 \\
\abbr{BG}-Dosage-Lethality & 2289 & 1776 & 2.58 & 392 & 4 & \best 9 & 12 & 9 & 8 & 56 \\
\abbr{BG}-Dosage-Rescue & 6444 & 3380 & 3.81 & 75 & 7 & \best 25 & 26 & 25 & 11 & 110 \\
\abbr{BG}-\abbr{D.}-Melanogaster & 60556 & 9330 & 12.98 & 303 & 83 & \best 198 & 291 & 198 & 83 & 6806 \\
\abbr{BG}-\abbr{E.}-Nidulans-Fgsc-A4 & 62 & 64 & 1.94 & 44 & 2 & \best 2 & \best 2 & 2 & 2 & 2 \\
\abbr{BG}-\abbr{E.}-Coli-K12-Mg1655 & 1889 & 1273 & 2.97 & 58 & 5 & \best 16 & \best 16 & 16 & 10 & 90 \\
\abbr{BG}-\abbr{E.}-Coli-K12-W3110 & 181620 & 4063 & 89.40 & 1187 & 133 & \best 672 & 802 & 672 & 290 & 83810 \\
\abbr{BG}-Far-Western & 1089 & 1199 & 1.82 & 60 & 3 & \best 4 & 5 & 4 & 3 & 6 \\
\abbr{BG}-Fret & 2395 & 1700 & 2.82 & 51 & 19 & \best 28 & 30 & 28 & 24 & 552 \\
\abbr{BG}-Gallus-Gallus & 436 & 413 & 2.11 & 110 & 4 & 6 & \best 5 & 5 & 4 & 12 \\
\abbr{BG}-Glycine-Max & 39 & 44 & 1.77 & 13 & 2 & \best 2 & \best 2 & 2 & 2 & 2 \\
\abbr{BG}-Hepatitus-C-Virus & 134 & 136 & 1.97 & 133 & 1 & \best 1 & \best 1 & 1 & 1 & 0 \\
\abbr{BG}-Homo-Sapiens & 369767 & 24093 & 30.69 & 2882 & 71 & \best 928 & 1170 & 928 & 263 & 68906 \\
\abbr{BG}-\abbr{HHV}-1 & 208 & 178 & 2.34 & 40 & 3 & \best 4 & 6 & 4 & 3 & 6 \\
\abbr{BG}-\abbr{HHV}-4 & 326 & 323 & 2.02 & 154 & 2 & \best 3 & \best 3 & 3 & 2 & 2 \\
\abbr{BG}-\abbr{HHV}-5 & 107 & 121 & 1.77 & 27 & 1 & 2 & \best 1 & 1 & 1 & 0 \\
\abbr{BG}-\abbr{HHV}-8 & 691 & 716 & 1.93 & 119 & 3 & 5 & \best 4 & 4 & 3 & 6 \\
\abbr{BG}-\abbr{HIV}-1 & 1319 & 1138 & 2.32 & 324 & 3 & \best 8 & 10 & 8 & 6 & 30 \\
\abbr{BG}-\abbr{HIV}-2 & 15 & 19 & 1.58 & 6 & 1 & 2 & \best 1 & 1 & 1 & 0 \\
\abbr{BG}-\abbr{HPV}-16 & 186 & 173 & 2.15 & 93 & 2 & \best 2 & 3 & 2 & 2 & 2 \\
Cannes2013 & 835892 & 438089 & 3.82 & 15169 & 27 & 370 & \best 332 & 332 & 114 & 12882 \\
CoW-interstate & 319 & 182 & 3.51 & 25 & 4 & 10 & \best 7 & 7 & 7 & 42 \\
DNC-emails & 4384 & 1866 & 4.70 & 402 & 17 & 42 & \best 41 & 41 & 28 & 756 \\
EU-email-core & 16064 & 986 & 32.58 & 345 & 34 & \best 161 & 187 & 161 & 74 & 5402 \\
JDK\_dependency & 53658 & 6434 & 16.68 & 5923 & 65 & \best 110 & 199 & 110 & 76 & 5700 \\
JUNG-javax & 50290 & 6120 & 16.43 & 5655 & 65 & \best 111 & 199 & 111 & 76 & 5700 \\
NYClimateMarch2014 & 327080 & 102378 & 6.39 & 14687 & 34 & 386 & \best 357 & 357 & 161 & 25760 \\
NZ\_legal & 15739 & 2141 & 14.70 & 429 & 25 & \best 122 & \best 122 & 122 & 68 & 4556 \\
Noordin-terror-loc & 190 & 127 & 2.99 & 18 & 3 & \best 5 & 6 & 5 & 4 & 12 \\
Noordin-terror-orgas & 181 & 129 & 2.81 & 21 & 3 & 5 & \best 4 & 4 & 3 & 6 \\
Noordin-terror-relation & 251 & 70 & 7.17 & 28 & 11 & 13 & \best 11 & 11 & 11 & 110 \\
ODLIS & 16377 & 2900 & 11.29 & 592 & 12 & \best 67 & 86 & 67 & 38 & 1406 \\
Opsahl-forum & 7036 & 899 & 15.65 & 128 & 14 & \best 94 & 105 & 94 & 42 & 1722 \\
Opsahl-socnet & 13838 & 1899 & 14.57 & 255 & 20 & \best 126 & 138 & 126 & 61 & 3660 \\
StackOverflow-tags & 245 & 115 & 4.26 & 16 & 6 & 7 & \best 6 & 6 & 6 & 30 \\
Y2H\_union & 2705 & 1966 & 2.75 & 89 & 4 & \best 13 & 14 & 13 & 7 & 42 \\
Yeast & 7182 & 2361 & 6.08 & 66 & 6 & \best 44 & 48 & 44 & 18 & 306 \\
actor\_movies & 1470404 & 511463 & 5.75 & 646 & 14 & 521 & \best 431 & 431 & 105 & 10920 \\
advogato & 39285 & 5155 & 15.24 & 803 & 25 & \best 202 & 320 & 202 & 86 & 7310 \\
airlines & 1297 & 235 & 11.04 & 130 & 13 & \best 26 & \best 26 & 26 & 18 & 306 \\
american\_revolution & 160 & 141 & 2.27 & 59 & 3 & \best 4 & 5 & 4 & 3 & 6 \\
as-22july06 & 48436 & 22963 & 4.22 & 2390 & 25 & 66 & \best 61 & 61 & 44 & 1892 \\
as-skitter & 11095298 & 1696415 & 13.08 & 35455 & 111 & \best 1973 & 2742 & 1973 & 365 & 132860 \\
as20000102 & 12572 & 6474 & 3.88 & 1458 & 12 & \best 32 & \best 32 & 32 & 21 & 420 \\
autobahn & 478 & 374 & 2.56 & 5 & 2 & 5 & \best 4 & 4 & 3 & 6 \\
bahamas & 246291 & 219856 & 2.24 & 14902 & 6 & \best 13 & 14 & 13 & 8 & 56 \\
bergen & 272 & 53 & 10.26 & 32 & 9 & \best 13 & 14 & 13 & 12 & 132 \\
bitcoin-otc-negative & 3259 & 1606 & 4.06 & 227 & 16 & \best 34 & \best 34 & 34 & 21 & 420 \\
bitcoin-otc-positive & 18591 & 5573 & 6.67 & 788 & 20 & \best 104 & 119 & 104 & 50 & 2450 \\
bn-fly-\abbr{d.}\_medulla\_1 & 8911 & 1781 & 10.01 & 927 & 18 & \best 78 & 84 & 78 & 44 & 1892 \\
bn-mouse\_retina\_1 & 90811 & 1076 & 168.79 & 744 & 121 & \best 359 & 428 & 359 & 223 & 49506 \\
boards\_gender\_1m & 19993 & 4134 & 9.67 & 88 & 25 & \best 39 & 46 & 39 & 25 & 600 \\
boards\_gender\_2m & 5598 & 4220 & 2.65 & 45 & 4 & 25 & \best 13 & 13 & 7 & 42 \\
ca-CondMat & 93439 & 23133 & 8.08 & 279 & 25 & \best 83 & 85 & 83 & 30 & 870 \\
ca-GrQc & 14484 & 5241 & 5.53 & 81 & 43 & \best 43 & \best 43 & 43 & 43 & 1806 \\
ca-HepPh & 118489 & 12006 & 19.74 & 491 & 135 & \best 238 & 265 & 238 & 238 & 56406 \\
capitalist & 1071 & 139 & 15.41 & 91 & 19 & \best 31 & 32 & 31 & 21 & 420 \\
celegans & 2148 & 297 & 14.46 & 134 & 10 & \best 42 & 44 & 42 & 21 & 420 \\
chess & 55899 & 7301 & 15.31 & 181 & 29 & \best 237 & 253 & 237 & 88 & 7656 \\
chicago & 1298 & 1467 & 1.77 & 12 & 1 & 11 & \best 1 & 1 & 1 & 0 \\
cit-HepPh & 420877 & 34546 & 24.37 & 846 & 30 & \best 411 & 495 & 411 & 89 & 7832 \\
cit-HepTh & 352285 & 27769 & 25.37 & 2468 & 37 & \best 497 & 533 & 497 & 128 & 16256 \\
codeminer & 1015 & 724 & 2.80 & 55 & 4 & 9 & \best 8 & 8 & 5 & 20 \\
columbia-mobility & 4147 & 863 & 9.61 & 228 & 9 & 18 & \best 16 & 16 & 9 & 72 \\
columbia-social & 7724 & 863 & 17.90 & 545 & 18 & \best 33 & 36 & 33 & 19 & 342 \\
cora\_citation & 89157 & 23166 & 7.70 & 377 & 13 & \best 88 & 102 & 88 & 30 & 870 \\
countries & 624402 & 592414 & 2.11 & 110602 & 6 & 39 & \best 26 & 26 & 16 & 240 \\
cpan-authors & 2112 & 839 & 5.03 & 327 & 9 & \best 23 & 27 & 23 & 17 & 272 \\
deezer & 498202 & 54573 & 18.26 & 420 & 21 & \best 319 & 363 & 319 & 60 & 3540 \\
digg & 86312 & 30398 & 5.68 & 285 & 8 & \best 113 & 138 & 113 & 46 & 2070 \\
diseasome & 2738 & 1419 & 3.86 & 84 & 11 & \best 11 & 12 & 11 & 11 & 110 \\
dogster\_friendships & 8546581 & 426820 & 40.05 & 46505 & 135 & \best 3379 & 4137 & 3379 & 1016 & 1031240 \\
dolphins & 159 & 62 & 5.13 & 12 & 4 & \best 8 & 9 & 8 & 6 & 30 \\
dutch-textiles & 90 & 48 & 3.75 & 31 & 5 & \best 5 & \best 5 & 5 & 5 & 20 \\
ecoli-transcript & 578 & 423 & 2.73 & 74 & 3 & \best 7 & 9 & 7 & 5 & 20 \\
edinburgh\_\abbr{assoc.} & 297094 & 23132 & 25.69 & 1062 & 34 & \best 848 & 1233 & 848 & 197 & 38612 \\
email-Enron & 183831 & 36692 & 10.02 & 1383 & 43 & \best 339 & 400 & 339 & 145 & 20880 \\
escorts & 39044 & 16730 & 4.67 & 305 & 11 & 110 & \best 93 & 93 & 45 & 1980 \\
euroroad & 1417 & 1174 & 2.41 & 10 & 2 & \best 4 & 5 & 4 & 3 & 6 \\
eva-corporate & 6711 & 7253 & 1.85 & 552 & 3 & 10 & \best 5 & 5 & 4 & 12 \\
exnet-water & 2416 & 1893 & 2.55 & 10 & 2 & 5 & \best 4 & 4 & 3 & 6 \\
facebook-links & 817090 & 63731 & 25.64 & 1098 & 52 & \best 895 & 1197 & 895 & 191 & 36290 \\
foldoc & 91471 & 13356 & 13.70 & 728 & 12 & \best 80 & 226 & 80 & 36 & 1260 \\
foodweb-caribbean & 3313 & 492 & 13.47 & 196 & 13 & \best 33 & 45 & 33 & 23 & 506 \\
foodweb-otago & 832 & 141 & 11.80 & 45 & 14 & 38 & \best 36 & 36 & 23 & 506 \\
football & 613 & 115 & 10.66 & 12 & 8 & 30 & \best 25 & 25 & 11 & 110 \\
google+ & 39194 & 23628 & 3.32 & 2761 & 12 & 63 & \best 61 & 61 & 38 & 1406 \\
gowalla & 950327 & 196591 & 9.67 & 14730 & 51 & \best 547 & 645 & 547 & 202 & 40602 \\
haggle & 2124 & 274 & 15.50 & 101 & 39 & \best 40 & 42 & 40 & 40 & 1560 \\
hex & 930 & 331 & 5.62 & 6 & 3 & 11 & \best 5 & 5 & 4 & 12 \\
hypertext\_2009 & 2196 & 113 & 38.87 & 98 & 28 & \best 66 & 68 & 66 & 43 & 1806 \\
ia-email-univ & 5451 & 1133 & 9.62 & 71 & 11 & \best 54 & 63 & 54 & 21 & 420 \\
ia-infect-dublin & 2765 & 410 & 13.49 & 50 & 17 & \best 34 & 50 & 34 & 21 & 420 \\
ia-reality & 7680 & 6809 & 2.26 & 261 & 5 & 56 & \best 25 & 25 & 12 & 132 \\
infectious & 2765 & 410 & 13.49 & 50 & 17 & \best 35 & 50 & 35 & 21 & 420 \\
ingredients & 431654 & 4372 & 197.46 & 3426 & 136 & \best 822 & 884 & 822 & 475 & 225150 \\
iscas89-s1196 & 537 & 377 & 2.85 & 16 & 2 & \best 7 & \best 7 & 7 & 4 & 12 \\
iscas89-s1238 & 625 & 416 & 3.00 & 18 & 2 & \best 7 & \best 7 & 7 & 5 & 20 \\
iscas89-s13207 & 3406 & 2492 & 2.73 & 37 & 4 & \best 6 & 10 & 6 & 6 & 30 \\
iscas89-s1423 & 554 & 423 & 2.62 & 17 & 2 & 5 & \best 4 & 4 & 3 & 6 \\
iscas89-s1488 & 779 & 463 & 3.37 & 53 & 3 & \best 9 & \best 9 & 9 & 7 & 42 \\
iscas89-s1494 & 796 & 473 & 3.37 & 56 & 3 & \best 9 & \best 9 & 9 & 7 & 42 \\
iscas89-s15850 & 4004 & 3247 & 2.47 & 25 & 4 & 9 & \best 7 & 7 & 4 & 12 \\
iscas89-s208 & 67 & 61 & 2.20 & 8 & 2 & 3 & \best 2 & 2 & 2 & 2 \\
iscas89-s27 & 8 & 9 & 1.78 & 3 & 1 & \best 1 & \best 1 & 1 & 1 & 0 \\
iscas89-s298 & 131 & 92 & 2.85 & 11 & 2 & 5 & \best 4 & 4 & 3 & 6 \\
iscas89-s344 & 122 & 100 & 2.44 & 9 & 2 & \best 3 & 4 & 3 & 3 & 6 \\
iscas89-s349 & 127 & 102 & 2.49 & 9 & 2 & 4 & \best 3 & 3 & 3 & 6 \\
iscas89-s35932 & 15961 & 12515 & 2.55 & 1440 & 2 & \best 3 & 6 & 3 & 2 & 2 \\
iscas89-s382 & 168 & 116 & 2.90 & 18 & 2 & 5 & \best 4 & 4 & 4 & 12 \\
iscas89-s38417 & 10635 & 9500 & 2.24 & 39 & 4 & \best 8 & 10 & 8 & 6 & 30 \\
iscas89-s38584 & 12573 & 9193 & 2.74 & 54 & 4 & \best 11 & 29 & 11 & 7 & 42 \\
iscas89-s386 & 200 & 114 & 3.51 & 23 & 3 & 6 & \best 5 & 5 & 4 & 12 \\
iscas89-s400 & 182 & 121 & 3.01 & 19 & 2 & \best 5 & \best 5 & 5 & 4 & 12 \\
iscas89-s420 & 145 & 129 & 2.25 & 9 & 2 & \best 3 & \best 3 & 3 & 2 & 2 \\
iscas89-s444 & 206 & 134 & 3.07 & 19 & 2 & 5 & \best 4 & 4 & 4 & 12 \\
iscas89-s510 & 251 & 172 & 2.92 & 12 & 2 & \best 5 & 6 & 5 & 4 & 12 \\
iscas89-s526 & 270 & 160 & 3.38 & 12 & 3 & \best 7 & \best 7 & 7 & 4 & 12 \\
iscas89-s526n & 268 & 159 & 3.37 & 12 & 3 & \best 6 & \best 6 & 6 & 4 & 12 \\
iscas89-s5378 & 1639 & 1411 & 2.32 & 10 & 3 & \best 6 & \best 6 & 6 & 5 & 20 \\
iscas89-s641 & 144 & 100 & 2.88 & 12 & 3 & \best 4 & \best 4 & 4 & 4 & 12 \\
iscas89-s713 & 180 & 137 & 2.63 & 12 & 3 & 6 & \best 4 & 4 & 4 & 12 \\
iscas89-s820 & 480 & 239 & 4.02 & 48 & 3 & \best 10 & 11 & 10 & 9 & 72 \\
iscas89-s832 & 498 & 245 & 4.07 & 49 & 3 & \best 11 & 13 & 11 & 9 & 72 \\
iscas89-s838 & 301 & 265 & 2.27 & 12 & 2 & \best 3 & \best 3 & 3 & 2 & 2 \\
iscas89-s9234 & 2370 & 1985 & 2.39 & 18 & 4 & \best 6 & 7 & 6 & 4 & 12 \\
iscas89-s953 & 454 & 332 & 2.73 & 12 & 2 & 7 & \best 6 & 6 & 3 & 6 \\
jazz & 2742 & 198 & 27.70 & 100 & 29 & \best 55 & 61 & 55 & 30 & 870 \\
karate & 78 & 34 & 4.59 & 17 & 4 & \best 5 & 6 & 5 & 4 & 12 \\
lederberg & 41532 & 8324 & 9.98 & 1103 & 15 & \best 116 & 132 & 116 & 47 & 2162 \\
lesmiserables & 254 & 77 & 6.60 & 36 & 9 & 10 & \best 9 & 9 & 9 & 72 \\
link-pedigree & 1125 & 898 & 2.51 & 14 & 2 & \best 4 & \best 4 & 4 & 2 & 2 \\
linux & 213217 & 30834 & 13.83 & 9338 & 23 & \best 197 & 206 & 197 & 106 & 11130 \\
livemocha & 2193083 & 104103 & 42.13 & 2980 & 92 & \best 2355 & 2825 & 2355 & 610 & 371490 \\
loc-brightkite\_edges & 214078 & 58228 & 7.35 & 1134 & 52 & \best 235 & 327 & 235 & 85 & 7140 \\
location & 293697 & 225486 & 2.61 & 12189 & 5 & 51 & \best 24 & 24 & 16 & 240 \\
mag\_geology\_coauthor & 4448428 & 2852295 & 3.12 & 1153 & 13 & 114 & \best 90 & 90 & 31 & 930 \\
marvel & 96662 & 19428 & 9.95 & 1625 & 18 & 141 & \best 135 & 135 & 58 & 3306 \\
mg\_casino & 326 & 109 & 5.98 & 94 & 9 & \best 9 & \best 9 & 9 & 9 & 72 \\
mg\_forrestgump & 271 & 94 & 5.77 & 89 & 8 & \best 8 & \best 8 & 8 & 8 & 56 \\
mg\_godfatherII & 219 & 78 & 5.62 & 34 & 8 & 9 & \best 8 & 8 & 8 & 56 \\
mg\_watchmen & 201 & 76 & 5.29 & 33 & 7 & \best 7 & \best 7 & 7 & 7 & 42 \\
minnesota & 3303 & 2642 & 2.50 & 5 & 2 & 5 & \best 4 & 4 & 3 & 6 \\
moreno\_health & 10455 & 2539 & 8.24 & 27 & 7 & 36 & \best 34 & 34 & 12 & 132 \\
mousebrain & 16089 & 213 & 151.07 & 205 & 111 & \best 189 & 193 & 189 & 141 & 19740 \\
movielens\_1m & 1000209 & 9746 & 205.26 & 3428 & 135 & \best 1394 & 1718 & 1394 & 554 & 306362 \\
movies & 192 & 101 & 3.80 & 19 & 3 & 9 & \best 8 & 8 & 5 & 20 \\
muenchen-bahn & 578 & 447 & 2.59 & 13 & 2 & \best 4 & 5 & 4 & 3 & 6 \\
munin & 1397 & 1324 & 2.11 & 66 & 3 & \best 5 & \best 5 & 5 & 3 & 6 \\
netscience & 2742 & 1461 & 3.75 & 34 & 19 & \best 19 & \best 19 & 19 & 19 & 342 \\
offshore & 505965 & 278877 & 3.63 & 37336 & 13 & \best 51 & 83 & 51 & 20 & 380 \\
openflights & 15677 & 2939 & 10.67 & 242 & 28 & 89 & \best 82 & 82 & 52 & 2652 \\
p2p-Gnutella04 & 39994 & 10876 & 7.35 & 103 & 7 & 60 & \best 58 & 58 & 23 & 506 \\
panama & 702437 & 556686 & 2.52 & 7015 & 62 & \best 64 & 253 & 64 & 62 & 3782 \\
paradise & 794545 & 542102 & 2.93 & 35359 & 23 & \best 117 & 158 & 117 & 55 & 2970 \\
photoviz\_dynamic & 610 & 376 & 3.24 & 29 & 4 & 10 & \best 9 & 9 & 7 & 42 \\
pigs & 592 & 492 & 2.41 & 39 & 2 & \best 4 & \best 4 & 4 & 3 & 6 \\
polblogs & 16715 & 1224 & 27.31 & 351 & 36 & \best 155 & 199 & 155 & 72 & 5112 \\
polbooks & 441 & 105 & 8.40 & 25 & 6 & \best 12 & \best 12 & 12 & 9 & 72 \\
pollination-carlinville & 15255 & 1500 & 20.34 & 157 & 18 & 163 & \best 141 & 141 & 53 & 2756 \\
pollination-daphni & 2933 & 797 & 7.36 & 124 & 9 & 45 & \best 37 & 37 & 26 & 650 \\
pollination-tenerife & 129 & 68 & 3.79 & 17 & 4 & 8 & \best 7 & 7 & 6 & 30 \\
pollination-uk & 16712 & 984 & 33.97 & 256 & 35 & \best 128 & 146 & 128 & 76 & 5700 \\
ratbrain & 23030 & 503 & 91.57 & 497 & 67 & \best 101 & 104 & 101 & 78 & 6006 \\
reactome & 147547 & 6327 & 46.64 & 855 & 62 & 239 & \best 234 & 234 & 184 & 33672 \\
residence\_hall & 1839 & 217 & 16.95 & 56 & 11 & \best 44 & 48 & 44 & 21 & 420 \\
rhesusbrain & 3054 & 242 & 25.24 & 111 & 19 & \best 65 & 67 & 65 & 37 & 1332 \\
roget-thesaurus & 3648 & 1010 & 7.22 & 28 & 6 & \best 26 & 27 & 26 & 11 & 110 \\
seventh-graders & 250 & 29 & 17.24 & 28 & 13 & \best 18 & 19 & 18 & 16 & 240 \\
slashdot\_threads & 117378 & 51083 & 4.60 & 2915 & 13 & \best 194 & 259 & 194 & 74 & 5402 \\
soc-Epinions1 & 405740 & 75879 & 10.69 & 3044 & 67 & \best 862 & 896 & 862 & 268 & 71556 \\
soc-Slashdot0811 & 469180 & 77360 & 12.13 & 2539 & 54 & \best 887 & 1028 & 887 & 232 & 53592 \\
soc-advogato & 39432 & 5167 & 15.26 & 807 & 25 & \best 205 & 316 & 205 & 86 & 7310 \\
soc-gplus & 39194 & 23628 & 3.32 & 2761 & 12 & 63 & \best 61 & 61 & 38 & 1406 \\
soc-hamsterster & 16630 & 2426 & 13.71 & 273 & 24 & \best 109 & 121 & 109 & 51 & 2550 \\
soc-wiki-Vote & 2914 & 889 & 6.56 & 102 & 9 & 28 & \best 26 & 26 & 16 & 240 \\
\abbr{sp\_d.sc.d.\_2} & 5539 & 238 & 46.55 & 88 & 33 & \best 114 & 117 & 114 & 57 & 3192 \\
teams & 1366466 & 935591 & 2.92 & 2671 & 9 & \best 258 & 290 & 258 & 127 & 16002 \\
train\_bombing & 243 & 64 & 7.59 & 29 & 10 & \best 10 & \best 10 & 10 & 10 & 90 \\
tv\_tropes & 3232134 & 152093 & 42.50 & 12400 & 115 & \best 2916 & 3715 & 2916 & 655 & 428370 \\
twittercrawl & 154824 & 3656 & 84.70 & 1084 & 132 & 636 & \best 598 & 598 & 237 & 55932 \\
ukroad & 15641 & 12378 & 2.53 & 5 & 3 & 5 & \best 4 & 4 & 3 & 6 \\
unicode\_languages & 1255 & 868 & 2.89 & 141 & 4 & 12 & \best 10 & 10 & 7 & 42 \\
wafa-ceos & 93 & 26 & 7.15 & 22 & 5 & \best 8 & \best 8 & 8 & 7 & 42 \\
wafa-eies & 652 & 45 & 28.98 & 44 & 24 & 32 & \best 29 & 29 & 27 & 702 \\
wafa-hightech & 159 & 21 & 15.14 & 20 & 12 & \best 15 & \best 15 & 15 & 13 & 156 \\
wafa-padgett & 27 & 15 & 3.60 & 8 & 3 & \best 4 & \best 4 & 4 & 3 & 6 \\
web-EPA & 8909 & 4271 & 4.17 & 175 & 6 & 47 & \best 42 & 42 & 16 & 240 \\
web-california & 15969 & 6175 & 5.17 & 199 & 11 & 69 & \best 63 & 63 & 26 & 650 \\
web-google & 2773 & 1299 & 4.27 & 59 & 17 & \best 17 & \best 17 & 17 & 17 & 272 \\
wiki-vote & 100762 & 7115 & 28.32 & 1065 & 53 & \best 456 & 548 & 456 & 162 & 26082 \\
wikipedia-norm & 15372 & 1881 & 16.34 & 455 & 22 & \best 107 & 122 & 107 & 59 & 3422 \\
win95pts & 112 & 99 & 2.26 & 9 & 2 & 4 & \best 3 & 3 & 3 & 6 \\
windsurfers & 336 & 43 & 15.63 & 31 & 11 & \best 21 & \best 21 & 21 & 15 & 210 \\
word\_adjacencies & 425 & 112 & 7.59 & 49 & 6 & \best 17 & 19 & 17 & 11 & 110 \\
zewail & 54182 & 6651 & 16.29 & 331 & 18 & \best 142 & 196 & 142 & 55 & 2970 \\
\end{longtable}

%% file: tables/queries.tex
\begin{longtable}{@{}l@{}S[table-format=7.0]S[table-format=7.0]S[table-format=4.2]S[table-format=4.0]S[table-format=3.0]S[table-format=4.0]RRRSSSRRRSSS@{}}
 &  &  &  &  &  &  & \multicolumn{3}{C}{{$\log n$}} & \multicolumn{3}{c}{{10}} & \multicolumn{3}{C}{{50}} & \multicolumn{3}{c}{{$\sqrt{n}$}} \\
{Network} & {$m$} & {$n$} & {$\bar d$} & {$\Delta$} & {deg} & {$\scol_2$}& {basel.} & {$s_2$ init} & {$s_2$ query}& {basel.} & {$s_2$ init} & {$s_2$ query}& {basel.} & {$s_2$ init} & {$s_2$ query}& {basel.} & {$s_2$ init} & {$s_2$ query}\\
\midrule
\endhead
\bottomrule
\endfoot
AS-oregon-1 & 23409 & 11174 & 4.19 & 2389 & 17 & 42 & 0.02 & 0.01 & 0.07 & 0.02 & 0.01 & 0.05 & 0.07 & 0.01 & 0.23 & 0.16 & 0.01 & 0.51 \\
AS-oregon-2 & 32730 & 11461 & 5.71 & 2432 & 31 & 83 & 0.03 & 0.01 & 0.11 & 0.02 & 0.02 & 0.07 & 0.09 & 0.01 & 0.33 & 0.19 & 0.01 & 0.78 \\
\abbr{BG}-\abbr{AC}-\abbr{Lumin.} & 2312 & 1840 & 2.51 & 376 & 6 & 11 & 0.01 & 0.00 & 0.03 & 0.01 & 0.00 & 0.03 & 0.05 & 0.00 & 0.14 & 0.04 & 0.00 & 0.12 \\
\abbr{BG}-\abbr{AC}-Ms & 321887 & 40495 & 15.90 & 2217 & 58 & 473 & 0.06 & 0.41 & 0.77 & 0.04 & 0.41 & 0.65 & 0.20 & 0.40 & 2.04 & 0.83 & 0.41 & 6.94 \\
\abbr{BG}-\abbr{AC}-Rna & 42815 & 13765 & 6.22 & 3572 & 54 & 124 & 0.03 & 0.02 & 0.09 & 0.02 & 0.02 & 0.06 & 0.08 & 0.02 & 0.23 & 0.16 & 0.02 & 0.48 \\
\abbr{BG}-\abbr{AC}-Western & 64046 & 21028 & 6.09 & 535 & 17 & 139 & 0.03 & 0.04 & 0.13 & 0.02 & 0.04 & 0.10 & 0.10 & 0.04 & 0.42 & 0.26 & 0.04 & 1.05 \\
\abbr{BG}-All & 1316843 & 75550 & 34.86 & 3620 & 134 & 1779 & 0.12 & 3.72 & 5.22 & 0.07 & 3.90 & 4.85 & 0.40 & 3.77 & 10.71 & 1.83 & 3.81 & 25.37 \\
\abbr{BG}-\abbr{A.}-Thaliana-Columbia & 47916 & 10417 & 9.20 & 1341 & 26 & 129 & 0.04 & 0.03 & 0.19 & 0.02 & 0.03 & 0.12 & 0.13 & 0.03 & 0.57 & 0.27 & 0.03 & 1.21 \\
\abbr{BG}-Biochemical-Activity & 17746 & 8620 & 4.12 & 427 & 11 & 56 & 0.05 & 0.01 & 0.07 & 0.02 & 0.01 & 0.05 & 0.08 & 0.01 & 0.24 & 0.13 & 0.01 & 0.45 \\
\abbr{BG}-Bos-Taurus & 424 & 454 & 1.87 & 27 & 3 & 5 & 0.01 & 0.00 & 0.02 & 0.01 & 0.00 & 0.02 & 0.05 & 0.00 & 0.13 & 0.02 & 0.00 & 0.05 \\
\abbr{BG}-\abbr{C.}-Elegans & 23646 & 6394 & 7.40 & 522 & 64 & 90 & 0.03 & 0.01 & 0.12 & 0.02 & 0.01 & 0.08 & 0.11 & 0.01 & 0.44 & 0.16 & 0.01 & 0.64 \\
\abbr{BG}-\abbr{C.}-Albicans-Sc5314 & 1609 & 1121 & 2.87 & 427 & 9 & 10 & 0.01 & 0.00 & 0.03 & 0.01 & 0.00 & 0.03 & 0.05 & 0.00 & 0.14 & 0.03 & 0.00 & 0.08 \\
\abbr{BG}-Canis-Familiaris & 125 & 143 & 1.75 & 90 & 2 & 2 & 0.01 & 0.00 & 0.02 & 0.01 & 0.00 & 0.02 & 0.04 & 0.00 & 0.11 & 0.01 & 0.00 & 0.03 \\
\abbr{BG}-Chemicals & 28093 & 33266 & 1.69 & 413 & 1 & 1 & 0.02 & 0.01 & 0.06 & 0.01 & 0.01 & 0.04 & 0.05 & 0.01 & 0.15 & 0.15 & 0.01 & 0.65 \\
\abbr{BG}-Co-Crystal-Structure & 2021 & 2291 & 1.76 & 92 & 5 & 5 & 0.01 & 0.00 & 0.03 & 0.01 & 0.00 & 0.03 & 0.05 & 0.00 & 0.15 & 0.04 & 0.00 & 0.13 \\
\abbr{BG}-Co-Fractionation & 56354 & 11017 & 10.23 & 187 & 83 & 146 & 0.04 & 0.04 & 0.21 & 0.03 & 0.05 & 0.15 & 0.14 & 0.04 & 0.66 & 0.30 & 0.04 & 1.49 \\
\abbr{BG}-Co-Localization & 4452 & 3543 & 2.51 & 63 & 6 & 13 & 0.01 & 0.00 & 0.04 & 0.01 & 0.00 & 0.03 & 0.06 & 0.00 & 0.17 & 0.06 & 0.00 & 0.19 \\
\abbr{BG}-Co-Purification & 5970 & 4326 & 2.76 & 1972 & 12 & 16 & 0.02 & 0.00 & 0.04 & 0.01 & 0.00 & 0.03 & 0.06 & 0.00 & 0.15 & 0.07 & 0.00 & 0.20 \\
\abbr{BG}-Cricetulus-Griseus & 57 & 69 & 1.65 & 30 & 1 & 1 & 0.01 & 0.00 & 0.01 & 0.01 & 0.00 & 0.02 & 0.04 & 0.00 & 0.11 & 0.01 & 0.00 & 0.02 \\
\abbr{BG}-Danio-Rerio & 266 & 261 & 2.04 & 61 & 3 & 3 & 0.01 & 0.00 & 0.02 & 0.01 & 0.00 & 0.02 & 0.05 & 0.00 & 0.13 & 0.02 & 0.00 & 0.04 \\
\abbr{BG}-\abbr{D.}-Discoideum-Ax4 & 20 & 27 & 1.48 & 4 & 1 & 1 & 0.00 & 0.00 & 0.02 & 0.01 & 0.00 & 0.02 & {---} & {---} & {---} & 0.01 & 0.00 & 0.01 \\
\abbr{BG}-Dosage-Growth-Defect & 2193 & 1447 & 3.03 & 213 & 5 & 15 & 0.01 & 0.00 & 0.03 & 0.01 & 0.00 & 0.03 & 0.06 & 0.00 & 0.16 & 0.04 & 0.00 & 0.11 \\
\abbr{BG}-Dosage-Lethality & 2289 & 1776 & 2.58 & 392 & 4 & 9 & 0.01 & 0.00 & 0.03 & 0.01 & 0.00 & 0.03 & 0.05 & 0.00 & 0.14 & 0.04 & 0.00 & 0.11 \\
\abbr{BG}-Dosage-Rescue & 6444 & 3380 & 3.81 & 75 & 7 & 25 & 0.02 & 0.00 & 0.05 & 0.01 & 0.00 & 0.04 & 0.07 & 0.00 & 0.21 & 0.08 & 0.00 & 0.22 \\
\abbr{BG}-\abbr{D.}-Melanogaster & 60556 & 9330 & 12.98 & 303 & 83 & 198 & 0.05 & 0.06 & 0.35 & 0.03 & 0.05 & 0.19 & 0.17 & 0.05 & 0.89 & 0.33 & 0.05 & 1.68 \\
\abbr{BG}-\abbr{E.}-Nidulans-Fgsc-A4 & 62 & 64 & 1.94 & 44 & 2 & 2 & 0.01 & 0.00 & 0.01 & 0.01 & 0.00 & 0.02 & 0.04 & 0.00 & 0.11 & 0.01 & 0.00 & 0.02 \\
\abbr{BG}-\abbr{E.}-Coli-K12-Mg1655 & 1889 & 1273 & 2.97 & 58 & 5 & 16 & 0.01 & 0.00 & 0.04 & 0.01 & 0.00 & 0.03 & 0.06 & 0.00 & 0.17 & 0.04 & 0.00 & 0.11 \\
\abbr{BG}-\abbr{E.}-Coli-K12-W3110 & 181620 & 4063 & 89.40 & 1187 & 133 & 672 & 0.18 & 0.36 & 2.40 & 0.15 & 0.36 & 1.82 & 0.58 & 0.36 & 7.56 & 0.63 & 0.35 & 6.82 \\
\abbr{BG}-Far-Western & 1089 & 1199 & 1.82 & 60 & 3 & 4 & 0.01 & 0.00 & 0.03 & 0.01 & 0.00 & 0.03 & 0.05 & 0.00 & 0.14 & 0.03 & 0.00 & 0.09 \\
\abbr{BG}-Fret & 2395 & 1700 & 2.82 & 51 & 19 & 28 & 0.01 & 0.00 & 0.04 & 0.01 & 0.00 & 0.06 & 0.06 & 0.00 & 0.18 & 0.05 & 0.00 & 0.14 \\
\abbr{BG}-Gallus-Gallus & 436 & 413 & 2.11 & 110 & 4 & 5 & 0.01 & 0.00 & 0.02 & 0.01 & 0.00 & 0.02 & 0.05 & 0.00 & 0.12 & 0.02 & 0.00 & 0.05 \\
\abbr{BG}-Glycine-Max & 39 & 44 & 1.77 & 13 & 2 & 2 & 0.01 & 0.00 & 0.01 & 0.01 & 0.00 & 0.02 & {---} & {---} & {---} & 0.01 & 0.00 & 0.02 \\
\abbr{BG}-Hepatitus-C-Virus & 134 & 136 & 1.97 & 133 & 1 & 1 & 0.01 & 0.00 & 0.01 & 0.01 & 0.00 & 0.02 & 0.03 & 0.00 & 0.10 & 0.01 & 0.00 & 0.02 \\
\abbr{BG}-Homo-Sapiens & 369767 & 24093 & 30.69 & 2882 & 71 & 928 & 0.09 & 0.89 & 1.74 & 0.06 & 0.90 & 1.49 & 0.34 & 0.95 & 5.15 & 0.90 & 0.90 & 9.90 \\
\abbr{BG}-\abbr{HHV}-1 & 208 & 178 & 2.34 & 40 & 3 & 4 & 0.01 & 0.00 & 0.02 & 0.01 & 0.00 & 0.02 & 0.05 & 0.00 & 0.12 & 0.01 & 0.00 & 0.04 \\
\abbr{BG}-\abbr{HHV}-4 & 326 & 323 & 2.02 & 154 & 2 & 3 & 0.01 & 0.00 & 0.02 & 0.01 & 0.00 & 0.02 & 0.04 & 0.00 & 0.11 & 0.01 & 0.00 & 0.04 \\
\abbr{BG}-\abbr{HHV}-5 & 107 & 121 & 1.77 & 27 & 1 & 1 & 0.01 & 0.00 & 0.02 & 0.01 & 0.00 & 0.02 & 0.04 & 0.00 & 0.12 & 0.01 & 0.00 & 0.03 \\
\abbr{BG}-\abbr{HHV}-8 & 691 & 716 & 1.93 & 119 & 3 & 4 & 0.01 & 0.00 & 0.02 & 0.01 & 0.00 & 0.02 & 0.04 & 0.00 & 0.12 & 0.03 & 0.00 & 0.06 \\
\abbr{BG}-\abbr{HIV}-1 & 1319 & 1138 & 2.32 & 324 & 3 & 8 & 0.01 & 0.00 & 0.03 & 0.01 & 0.00 & 0.02 & 0.04 & 0.00 & 0.11 & 0.03 & 0.00 & 0.07 \\
\abbr{BG}-\abbr{HIV}-2 & 15 & 19 & 1.58 & 6 & 1 & 1 & 0.01 & 0.00 & 0.01 & 0.01 & 0.00 & 0.02 & {---} & {---} & {---} & 0.00 & 0.00 & 0.01 \\
\abbr{BG}-\abbr{HPV}-16 & 186 & 173 & 2.15 & 93 & 2 & 2 & 0.01 & 0.00 & 0.02 & 0.01 & 0.00 & 0.02 & 0.04 & 0.00 & 0.10 & 0.01 & 0.00 & 0.03 \\
Cannes2013 & 835892 & 438089 & 3.82 & 15169 & 27 & 332 & 0.04 & 0.65 & 0.88 & 0.02 & 0.67 & 0.81 & 0.11 & 0.65 & 1.23 & 1.29 & 0.69 & 8.74 \\
CoW-interstate & 319 & 182 & 3.51 & 25 & 4 & 7 & 0.01 & 0.00 & 0.02 & 0.01 & 0.00 & 0.03 & 0.06 & 0.00 & 0.16 & 0.02 & 0.00 & 0.05 \\
DNC-emails & 4384 & 1866 & 4.70 & 402 & 17 & 41 & 0.02 & 0.00 & 0.05 & 0.01 & 0.00 & 0.04 & 0.07 & 0.00 & 0.22 & 0.06 & 0.00 & 0.17 \\
EU-email-core & 16064 & 986 & 32.58 & 345 & 34 & 161 & 0.06 & 0.01 & 0.29 & 0.06 & 0.01 & 0.29 & 0.22 & 0.01 & 1.39 & 0.16 & 0.01 & 0.81 \\
JDK\_dependency & 53658 & 6434 & 16.68 & 5923 & 65 & 110 & 0.04 & 0.02 & 0.16 & 0.03 & 0.02 & 0.11 & 0.16 & 0.02 & 0.52 & 0.23 & 0.02 & 0.69 \\
JUNG-javax & 50290 & 6120 & 16.43 & 5655 & 65 & 111 & 0.04 & 0.02 & 0.15 & 0.03 & 0.02 & 0.10 & 0.16 & 0.02 & 0.53 & 0.22 & 0.02 & 0.61 \\
NYClimateMarch2014 & 327080 & 102378 & 6.39 & 14687 & 34 & 357 & 0.03 & 0.27 & 0.44 & 0.02 & 0.27 & 0.38 & 0.10 & 0.27 & 0.88 & 0.59 & 0.27 & 3.78 \\
NZ\_legal & 15739 & 2141 & 14.70 & 429 & 25 & 122 & 0.04 & 0.01 & 0.18 & 0.03 & 0.01 & 0.13 & 0.18 & 0.01 & 0.72 & 0.15 & 0.01 & 0.64 \\
Noordin-terror-loc & 190 & 127 & 2.99 & 18 & 3 & 5 & 0.01 & 0.00 & 0.02 & 0.01 & 0.00 & 0.03 & 0.05 & 0.00 & 0.15 & 0.02 & 0.00 & 0.04 \\
Noordin-terror-orgas & 181 & 129 & 2.81 & 21 & 3 & 4 & 0.01 & 0.00 & 0.02 & 0.01 & 0.00 & 0.03 & 0.05 & 0.00 & 0.14 & 0.01 & 0.00 & 0.03 \\
Noordin-terror-relation & 251 & 70 & 7.17 & 28 & 11 & 11 & 0.01 & 0.00 & 0.03 & 0.02 & 0.00 & 0.04 & 0.06 & 0.00 & 0.20 & 0.02 & 0.00 & 0.04 \\
ODLIS & 16377 & 2900 & 11.29 & 592 & 12 & 67 & 0.03 & 0.01 & 0.13 & 0.03 & 0.01 & 0.10 & 0.14 & 0.01 & 0.54 & 0.15 & 0.01 & 0.64 \\
Opsahl-forum & 7036 & 899 & 15.65 & 128 & 14 & 94 & 0.07 & 0.00 & 0.13 & 0.04 & 0.00 & 0.14 & 0.16 & 0.00 & 0.70 & 0.10 & 0.00 & 0.36 \\
Opsahl-socnet & 13838 & 1899 & 14.57 & 255 & 20 & 126 & 0.04 & 0.01 & 0.16 & 0.03 & 0.01 & 0.14 & 0.16 & 0.01 & 0.76 & 0.14 & 0.01 & 0.60 \\
StackOverflow-tags & 245 & 115 & 4.26 & 16 & 6 & 6 & 0.01 & 0.00 & 0.02 & 0.02 & 0.00 & 0.03 & 0.06 & 0.00 & 0.17 & 0.02 & 0.00 & 0.04 \\
Y2H\_union & 2705 & 1966 & 2.75 & 89 & 4 & 13 & 0.01 & 0.00 & 0.03 & 0.01 & 0.00 & 0.03 & 0.06 & 0.00 & 0.17 & 0.05 & 0.00 & 0.14 \\
Yeast & 7182 & 2361 & 6.08 & 66 & 6 & 44 & 0.02 & 0.00 & 0.06 & 0.02 & 0.00 & 0.05 & 0.09 & 0.00 & 0.29 & 0.09 & 0.00 & 0.28 \\
actor\_movies & 1470404 & 511463 & 5.75 & 646 & 14 & 431 & 0.06 & 1.60 & 2.01 & 0.03 & 1.65 & 1.94 & 0.15 & 1.56 & 2.51 & 2.06 & 1.60 & 15.20 \\
advogato & 39285 & 5155 & 15.24 & 803 & 25 & 202 & 0.05 & 0.03 & 0.28 & 0.03 & 0.03 & 0.19 & 0.18 & 0.03 & 1.01 & 0.23 & 0.03 & 1.09 \\
airlines & 1297 & 235 & 11.04 & 130 & 13 & 26 & 0.02 & 0.00 & 0.04 & 0.02 & 0.00 & 0.06 & 0.09 & 0.00 & 0.27 & 0.03 & 0.00 & 0.09 \\
american\_revolution & 160 & 141 & 2.27 & 59 & 3 & 4 & 0.01 & 0.00 & 0.02 & 0.01 & 0.00 & 0.02 & 0.04 & 0.00 & 0.11 & 0.01 & 0.00 & 0.03 \\
as-22july06 & 48436 & 22963 & 4.22 & 2390 & 25 & 61 & 0.02 & 0.02 & 0.08 & 0.02 & 0.02 & 0.06 & 0.08 & 0.02 & 0.27 & 0.21 & 0.02 & 0.75 \\
as-skitter & 11095298 & 1696415 & 13.08 & 35455 & 111 & 1973 & 0.15 & 20.11 & 21.97 & 0.06 & 19.98 & 21.46 & 0.29 & 21.04 & 25.06 & 7.34 & 20.13 & 81.39 \\
as20000102 & 12572 & 6474 & 3.88 & 1458 & 12 & 32 & 0.02 & 0.00 & 0.05 & 0.01 & 0.00 & 0.04 & 0.07 & 0.00 & 0.20 & 0.10 & 0.00 & 0.30 \\
autobahn & 478 & 374 & 2.56 & 5 & 2 & 4 & 0.01 & 0.00 & 0.02 & 0.01 & 0.00 & 0.03 & 0.05 & 0.00 & 0.15 & 0.02 & 0.00 & 0.06 \\
bahamas & 246291 & 219856 & 2.24 & 14902 & 6 & 13 & 0.02 & 0.11 & 0.22 & 0.02 & 0.12 & 0.18 & 0.07 & 0.12 & 0.40 & 0.51 & 0.11 & 3.54 \\
bergen & 272 & 53 & 10.26 & 32 & 9 & 13 & 0.01 & 0.00 & 0.03 & 0.02 & 0.00 & 0.05 & 0.06 & 0.00 & 0.23 & 0.02 & 0.00 & 0.04 \\
bitcoin-otc-negative & 3259 & 1606 & 4.06 & 227 & 16 & 34 & 0.02 & 0.00 & 0.04 & 0.01 & 0.00 & 0.04 & 0.07 & 0.00 & 0.20 & 0.05 & 0.00 & 0.15 \\
bitcoin-otc-positive & 18591 & 5573 & 6.67 & 788 & 20 & 104 & 0.03 & 0.01 & 0.10 & 0.02 & 0.01 & 0.07 & 0.11 & 0.01 & 0.39 & 0.14 & 0.01 & 0.49 \\
bn-fly-\abbr{d.}\_medulla\_1 & 8911 & 1781 & 10.01 & 927 & 18 & 78 & 0.03 & 0.00 & 0.10 & 0.02 & 0.00 & 0.09 & 0.12 & 0.00 & 0.49 & 0.10 & 0.00 & 0.37 \\
bn-mouse\_retina\_1 & 90811 & 1076 & 168.79 & 744 & 121 & 359 & 0.21 & 0.11 & 1.90 & 0.19 & 0.11 & 1.75 & 0.56 & 0.11 & 6.27 & 0.42 & 0.11 & 4.16 \\
boards\_gender\_1m & 19993 & 4134 & 9.67 & 88 & 25 & 39 & 0.04 & 0.01 & 0.10 & 0.03 & 0.01 & 0.07 & 0.13 & 0.01 & 0.40 & 0.16 & 0.01 & 0.41 \\
boards\_gender\_2m & 5598 & 4220 & 2.65 & 45 & 4 & 13 & 0.02 & 0.00 & 0.04 & 0.01 & 0.00 & 0.03 & 0.06 & 0.00 & 0.16 & 0.07 & 0.00 & 0.19 \\
ca-CondMat & 93439 & 23133 & 8.08 & 279 & 25 & 83 & 0.03 & 0.04 & 0.14 & 0.02 & 0.04 & 0.11 & 0.12 & 0.04 & 0.47 & 0.37 & 0.04 & 1.19 \\
ca-GrQc & 14484 & 5241 & 5.53 & 81 & 43 & 43 & 0.02 & 0.01 & 0.07 & 0.02 & 0.01 & 0.05 & 0.09 & 0.01 & 0.26 & 0.12 & 0.01 & 0.34 \\
ca-HepPh & 118489 & 12006 & 19.74 & 491 & 135 & 238 & 0.06 & 0.08 & 0.38 & 0.04 & 0.08 & 0.25 & 0.21 & 0.08 & 1.07 & 0.43 & 0.09 & 2.32 \\
capitalist & 1071 & 139 & 15.41 & 91 & 19 & 31 & 0.02 & 0.00 & 0.06 & 0.03 & 0.00 & 0.08 & 0.09 & 0.00 & 0.38 & 0.03 & 0.00 & 0.13 \\
celegans & 2148 & 297 & 14.46 & 134 & 10 & 42 & 0.03 & 0.00 & 0.09 & 0.03 & 0.00 & 0.10 & 0.12 & 0.00 & 0.49 & 0.05 & 0.00 & 0.17 \\
chess & 55899 & 7301 & 15.31 & 181 & 29 & 237 & 0.05 & 0.05 & 0.31 & 0.04 & 0.05 & 0.21 & 0.18 & 0.05 & 1.21 & 0.35 & 0.05 & 1.60 \\
chicago & 1298 & 1467 & 1.77 & 12 & 1 & 1 & 0.01 & 0.00 & 0.03 & 0.01 & 0.00 & 0.02 & 0.04 & 0.00 & 0.13 & 0.03 & 0.00 & 0.09 \\
cit-HepPh & 420877 & 34546 & 24.37 & 846 & 30 & 411 & 0.08 & 0.47 & 0.97 & 0.05 & 0.47 & 0.80 & 0.29 & 0.47 & 2.62 & 1.10 & 0.47 & 7.98 \\
cit-HepTh & 352285 & 27769 & 25.37 & 2468 & 37 & 497 & 0.08 & 0.46 & 0.94 & 0.05 & 0.43 & 0.82 & 0.30 & 0.44 & 2.72 & 0.90 & 0.44 & 6.67 \\
codeminer & 1015 & 724 & 2.80 & 55 & 4 & 8 & 0.01 & 0.00 & 0.03 & 0.01 & 0.00 & 0.03 & 0.06 & 0.00 & 0.15 & 0.03 & 0.00 & 0.08 \\
columbia-mobility & 4147 & 863 & 9.61 & 228 & 9 & 16 & 0.02 & 0.00 & 0.06 & 0.02 & 0.00 & 0.07 & 0.11 & 0.00 & 0.33 & 0.07 & 0.00 & 0.17 \\
columbia-social & 7724 & 863 & 17.90 & 545 & 18 & 33 & 0.04 & 0.00 & 0.09 & 0.10 & 0.00 & 0.10 & 0.15 & 0.00 & 0.48 & 0.09 & 0.00 & 0.25 \\
cora\_citation & 89157 & 23166 & 7.70 & 377 & 13 & 88 & 0.03 & 0.05 & 0.15 & 0.02 & 0.05 & 0.12 & 0.12 & 0.04 & 0.47 & 0.34 & 0.04 & 1.23 \\
countries & 624402 & 592414 & 2.11 & 110602 & 6 & 26 & 0.04 & 0.36 & 0.51 & 0.02 & 0.34 & 0.43 & 0.09 & 0.41 & 0.71 & 1.13 & 0.39 & 7.49 \\
cpan-authors & 2112 & 839 & 5.03 & 327 & 9 & 23 & 0.01 & 0.00 & 0.04 & 0.02 & 0.00 & 0.04 & 0.07 & 0.00 & 0.18 & 0.04 & 0.00 & 0.09 \\
deezer & 498202 & 54573 & 18.26 & 420 & 21 & 319 & 0.07 & 0.54 & 0.95 & 0.04 & 0.54 & 0.81 & 0.24 & 0.54 & 2.24 & 1.02 & 0.56 & 6.59 \\
digg & 86312 & 30398 & 5.68 & 285 & 8 & 113 & 0.03 & 0.06 & 0.15 & 0.02 & 0.06 & 0.13 & 0.10 & 0.06 & 0.43 & 0.32 & 0.06 & 1.45 \\
diseasome & 2738 & 1419 & 3.86 & 84 & 11 & 11 & 0.02 & 0.00 & 0.04 & 0.01 & 0.00 & 0.03 & 0.10 & 0.00 & 0.18 & 0.05 & 0.00 & 0.13 \\
dogster\_friendships & 8546581 & 426820 & 40.05 & 46505 & 135 & 3379 & 0.18 & 69.33 & 78.40 & 0.09 & 69.27 & 76.17 & 0.53 & 68.69 & 88.64 & 5.92 & 69.27 & 232.92 \\
dolphins & 159 & 62 & 5.13 & 12 & 4 & 8 & 0.01 & 0.00 & 0.02 & 0.02 & 0.00 & 0.04 & 0.05 & 0.00 & 0.18 & 0.01 & 0.00 & 0.03 \\
dutch-textiles & 90 & 48 & 3.75 & 31 & 5 & 5 & 0.01 & 0.00 & 0.02 & 0.01 & 0.00 & 0.03 & {---} & {---} & {---} & 0.01 & 0.00 & 0.02 \\
ecoli-transcript & 578 & 423 & 2.73 & 74 & 3 & 7 & 0.01 & 0.00 & 0.02 & 0.01 & 0.00 & 0.03 & 0.05 & 0.00 & 0.14 & 0.02 & 0.00 & 0.05 \\
edinburgh\_\abbr{assoc.} & 297094 & 23132 & 25.69 & 1062 & 34 & 848 & 0.08 & 0.74 & 1.41 & 0.05 & 0.84 & 1.33 & 0.30 & 0.74 & 3.78 & 0.77 & 0.74 & 7.79 \\
email-Enron & 183831 & 36692 & 10.02 & 1383 & 43 & 339 & 0.04 & 0.17 & 0.36 & 0.03 & 0.17 & 0.30 & 0.14 & 0.17 & 1.05 & 0.50 & 0.17 & 3.33 \\
escorts & 39044 & 16730 & 4.67 & 305 & 11 & 93 & 0.02 & 0.02 & 0.09 & 0.02 & 0.02 & 0.07 & 0.08 & 0.02 & 0.33 & 0.19 & 0.02 & 0.71 \\
euroroad & 1417 & 1174 & 2.41 & 10 & 2 & 4 & 0.01 & 0.00 & 0.03 & 0.01 & 0.00 & 0.03 & 0.06 & 0.00 & 0.15 & 0.03 & 0.00 & 0.10 \\
eva-corporate & 6711 & 7253 & 1.85 & 552 & 3 & 5 & 0.01 & 0.00 & 0.04 & 0.01 & 0.00 & 0.03 & 0.05 & 0.00 & 0.13 & 0.08 & 0.00 & 0.24 \\
exnet-water & 2416 & 1893 & 2.55 & 10 & 2 & 4 & 0.01 & 0.00 & 0.03 & 0.01 & 0.00 & 0.03 & 0.06 & 0.00 & 0.16 & 0.05 & 0.00 & 0.13 \\
facebook-links & 817090 & 63731 & 25.64 & 1098 & 52 & 895 & 0.09 & 1.66 & 2.49 & 0.06 & 1.72 & 2.27 & 0.31 & 1.67 & 5.16 & 1.36 & 1.66 & 14.18 \\
foldoc & 91471 & 13356 & 13.70 & 728 & 12 & 80 & 0.05 & 0.05 & 0.25 & 0.03 & 0.05 & 0.17 & 0.22 & 0.05 & 0.76 & 0.35 & 0.05 & 1.47 \\
foodweb-caribbean & 3313 & 492 & 13.47 & 196 & 13 & 33 & 0.02 & 0.00 & 0.06 & 0.03 & 0.00 & 0.07 & 0.11 & 0.00 & 0.34 & 0.06 & 0.00 & 0.16 \\
foodweb-otago & 832 & 141 & 11.80 & 45 & 14 & 36 & 0.02 & 0.00 & 0.06 & 0.02 & 0.00 & 0.07 & 0.08 & 0.00 & 0.32 & 0.03 & 0.00 & 0.08 \\
football & 613 & 115 & 10.66 & 12 & 8 & 25 & 0.02 & 0.00 & 0.05 & 0.03 & 0.00 & 0.07 & 0.08 & 0.00 & 0.36 & 0.03 & 0.00 & 0.08 \\
google+ & 39194 & 23628 & 3.32 & 2761 & 12 & 61 & 0.02 & 0.01 & 0.06 & 0.01 & 0.02 & 0.05 & 0.06 & 0.01 & 0.20 & 0.15 & 0.01 & 0.55 \\
gowalla & 950327 & 196591 & 9.67 & 14730 & 51 & 547 & 0.06 & 1.11 & 1.56 & 0.03 & 1.11 & 1.36 & 0.18 & 1.10 & 2.34 & 1.63 & 1.11 & 14.60 \\
haggle & 2124 & 274 & 15.50 & 101 & 39 & 40 & 0.02 & 0.00 & 0.06 & 0.02 & 0.00 & 0.07 & 0.09 & 0.00 & 0.32 & 0.04 & 0.00 & 0.11 \\
hex & 930 & 331 & 5.62 & 6 & 3 & 5 & 0.01 & 0.00 & 0.04 & 0.02 & 0.00 & 0.04 & 0.08 & 0.00 & 0.23 & 0.03 & 0.00 & 0.08 \\
hypertext\_2009 & 2196 & 113 & 38.87 & 98 & 28 & 66 & 0.03 & 0.00 & 0.16 & 0.04 & 0.00 & 0.22 & 0.13 & 0.00 & 1.19 & 0.04 & 0.00 & 0.24 \\
ia-email-univ & 5451 & 1133 & 9.62 & 71 & 11 & 54 & 0.03 & 0.00 & 0.09 & 0.03 & 0.00 & 0.08 & 0.12 & 0.00 & 0.46 & 0.08 & 0.00 & 0.26 \\
ia-infect-dublin & 2765 & 410 & 13.49 & 50 & 17 & 34 & 0.03 & 0.00 & 0.08 & 0.03 & 0.00 & 0.09 & 0.12 & 0.00 & 0.46 & 0.06 & 0.00 & 0.19 \\
ia-reality & 7680 & 6809 & 2.26 & 261 & 5 & 25 & 0.01 & 0.00 & 0.04 & 0.01 & 0.00 & 0.03 & 0.05 & 0.00 & 0.14 & 0.07 & 0.00 & 0.23 \\
infectious & 2765 & 410 & 13.49 & 50 & 17 & 35 & 0.03 & 0.00 & 0.08 & 0.03 & 0.00 & 0.09 & 0.12 & 0.00 & 0.46 & 0.06 & 0.00 & 0.19 \\
ingredients & 431654 & 4372 & 197.46 & 3426 & 136 & 822 & 0.33 & 0.85 & 5.17 & 0.24 & 0.85 & 3.63 & 0.93 & 0.85 & 11.12 & 1.10 & 0.85 & 12.26 \\
iscas89-s1196 & 537 & 377 & 2.85 & 16 & 2 & 7 & 0.01 & 0.00 & 0.03 & 0.01 & 0.00 & 0.03 & 0.06 & 0.00 & 0.15 & 0.02 & 0.00 & 0.06 \\
iscas89-s1238 & 625 & 416 & 3.00 & 18 & 2 & 7 & 0.01 & 0.00 & 0.03 & 0.01 & 0.00 & 0.03 & 0.06 & 0.00 & 0.16 & 0.03 & 0.00 & 0.06 \\
iscas89-s13207 & 3406 & 2492 & 2.73 & 37 & 4 & 6 & 0.01 & 0.00 & 0.04 & 0.01 & 0.00 & 0.03 & 0.06 & 0.00 & 0.17 & 0.06 & 0.00 & 0.17 \\
iscas89-s1423 & 554 & 423 & 2.62 & 17 & 2 & 4 & 0.01 & 0.00 & 0.02 & 0.01 & 0.00 & 0.03 & 0.06 & 0.00 & 0.15 & 0.02 & 0.00 & 0.06 \\
iscas89-s1488 & 779 & 463 & 3.37 & 53 & 3 & 9 & 0.01 & 0.00 & 0.03 & 0.01 & 0.00 & 0.03 & 0.06 & 0.00 & 0.15 & 0.03 & 0.00 & 0.06 \\
iscas89-s1494 & 796 & 473 & 3.37 & 56 & 3 & 9 & 0.01 & 0.00 & 0.03 & 0.01 & 0.00 & 0.03 & 0.06 & 0.00 & 0.15 & 0.03 & 0.00 & 0.06 \\
iscas89-s15850 & 4004 & 3247 & 2.47 & 25 & 4 & 7 & 0.01 & 0.00 & 0.04 & 0.01 & 0.00 & 0.03 & 0.06 & 0.00 & 0.16 & 0.06 & 0.00 & 0.17 \\
iscas89-s208 & 67 & 61 & 2.20 & 8 & 2 & 2 & 0.01 & 0.00 & 0.02 & 0.01 & 0.00 & 0.03 & 0.05 & 0.00 & 0.13 & 0.01 & 0.00 & 0.02 \\
iscas89-s27 & 8 & 9 & 1.78 & 3 & 1 & 1 & 0.00 & 0.00 & 0.01 & {---} & {---} & {---} & {---} & {---} & {---} & 0.00 & 0.00 & 0.01 \\
iscas89-s298 & 131 & 92 & 2.85 & 11 & 2 & 4 & 0.01 & 0.00 & 0.02 & 0.01 & 0.00 & 0.03 & 0.06 & 0.00 & 0.15 & 0.01 & 0.00 & 0.03 \\
iscas89-s344 & 122 & 100 & 2.44 & 9 & 2 & 3 & 0.01 & 0.00 & 0.02 & 0.01 & 0.00 & 0.03 & 0.06 & 0.00 & 0.14 & 0.01 & 0.00 & 0.03 \\
iscas89-s349 & 127 & 102 & 2.49 & 9 & 2 & 3 & 0.01 & 0.00 & 0.02 & 0.01 & 0.00 & 0.03 & 0.06 & 0.00 & 0.15 & 0.01 & 0.00 & 0.03 \\
iscas89-s35932 & 15961 & 12515 & 2.55 & 1440 & 2 & 3 & 0.02 & 0.01 & 0.05 & 0.01 & 0.01 & 0.04 & 0.06 & 0.01 & 0.17 & 0.14 & 0.01 & 0.45 \\
iscas89-s382 & 168 & 116 & 2.90 & 18 & 2 & 4 & 0.01 & 0.00 & 0.02 & 0.01 & 0.00 & 0.04 & 0.06 & 0.00 & 0.15 & 0.01 & 0.00 & 0.03 \\
iscas89-s38417 & 10635 & 9500 & 2.24 & 39 & 4 & 8 & 0.02 & 0.00 & 0.05 & 0.01 & 0.00 & 0.03 & 0.05 & 0.00 & 0.16 & 0.10 & 0.00 & 0.33 \\
iscas89-s38584 & 12573 & 9193 & 2.74 & 54 & 4 & 11 & 0.02 & 0.01 & 0.05 & 0.01 & 0.01 & 0.04 & 0.06 & 0.01 & 0.18 & 0.12 & 0.01 & 0.40 \\
iscas89-s386 & 200 & 114 & 3.51 & 23 & 3 & 5 & 0.01 & 0.00 & 0.02 & 0.01 & 0.00 & 0.03 & 0.06 & 0.00 & 0.15 & 0.01 & 0.00 & 0.03 \\
iscas89-s400 & 182 & 121 & 3.01 & 19 & 2 & 5 & 0.01 & 0.00 & 0.02 & 0.01 & 0.00 & 0.03 & 0.06 & 0.00 & 0.16 & 0.02 & 0.00 & 0.04 \\
iscas89-s420 & 145 & 129 & 2.25 & 9 & 2 & 3 & 0.01 & 0.00 & 0.02 & 0.01 & 0.00 & 0.03 & 0.05 & 0.00 & 0.14 & 0.01 & 0.00 & 0.03 \\
iscas89-s444 & 206 & 134 & 3.07 & 19 & 2 & 4 & 0.01 & 0.00 & 0.02 & 0.01 & 0.00 & 0.03 & 0.06 & 0.00 & 0.16 & 0.02 & 0.00 & 0.04 \\
iscas89-s510 & 251 & 172 & 2.92 & 12 & 2 & 5 & 0.01 & 0.00 & 0.02 & 0.01 & 0.00 & 0.03 & 0.06 & 0.00 & 0.16 & 0.02 & 0.00 & 0.05 \\
iscas89-s526 & 270 & 160 & 3.38 & 12 & 3 & 7 & 0.01 & 0.00 & 0.02 & 0.01 & 0.00 & 0.03 & 0.06 & 0.00 & 0.16 & 0.02 & 0.00 & 0.04 \\
iscas89-s526n & 268 & 159 & 3.37 & 12 & 3 & 6 & 0.01 & 0.00 & 0.02 & 0.01 & 0.00 & 0.03 & 0.06 & 0.00 & 0.16 & 0.02 & 0.00 & 0.04 \\
iscas89-s5378 & 1639 & 1411 & 2.32 & 10 & 3 & 6 & 0.01 & 0.00 & 0.03 & 0.01 & 0.00 & 0.03 & 0.05 & 0.00 & 0.15 & 0.04 & 0.00 & 0.11 \\
iscas89-s641 & 144 & 100 & 2.88 & 12 & 3 & 4 & 0.01 & 0.00 & 0.02 & 0.01 & 0.00 & 0.03 & 0.06 & 0.00 & 0.15 & 0.01 & 0.00 & 0.03 \\
iscas89-s713 & 180 & 137 & 2.63 & 12 & 3 & 4 & 0.01 & 0.00 & 0.02 & 0.01 & 0.00 & 0.03 & 0.05 & 0.00 & 0.15 & 0.01 & 0.00 & 0.04 \\
iscas89-s820 & 480 & 239 & 4.02 & 48 & 3 & 10 & 0.01 & 0.00 & 0.02 & 0.01 & 0.00 & 0.03 & 0.06 & 0.00 & 0.16 & 0.02 & 0.00 & 0.05 \\
iscas89-s832 & 498 & 245 & 4.07 & 49 & 3 & 11 & 0.01 & 0.00 & 0.02 & 0.01 & 0.00 & 0.03 & 0.06 & 0.00 & 0.16 & 0.02 & 0.00 & 0.05 \\
iscas89-s838 & 301 & 265 & 2.27 & 12 & 2 & 3 & 0.01 & 0.00 & 0.02 & 0.01 & 0.00 & 0.03 & 0.05 & 0.00 & 0.14 & 0.02 & 0.00 & 0.05 \\
iscas89-s9234 & 2370 & 1985 & 2.39 & 18 & 4 & 6 & 0.01 & 0.00 & 0.03 & 0.01 & 0.00 & 0.03 & 0.06 & 0.00 & 0.16 & 0.05 & 0.00 & 0.13 \\
iscas89-s953 & 454 & 332 & 2.73 & 12 & 2 & 6 & 0.01 & 0.00 & 0.03 & 0.01 & 0.00 & 0.03 & 0.06 & 0.00 & 0.15 & 0.02 & 0.00 & 0.06 \\
jazz & 2742 & 198 & 27.70 & 100 & 29 & 55 & 0.04 & 0.00 & 0.13 & 0.04 & 0.00 & 0.17 & 0.12 & 0.00 & 0.79 & 0.05 & 0.00 & 0.23 \\
karate & 78 & 34 & 4.59 & 17 & 4 & 5 & 0.01 & 0.00 & 0.02 & 0.01 & 0.00 & 0.03 & {---} & {---} & {---} & 0.01 & 0.00 & 0.02 \\
lederberg & 41532 & 8324 & 9.98 & 1103 & 15 & 116 & 0.04 & 0.02 & 0.18 & 0.03 & 0.02 & 0.11 & 0.14 & 0.02 & 0.56 & 0.24 & 0.02 & 0.95 \\
lesmiserables & 254 & 77 & 6.60 & 36 & 9 & 9 & 0.01 & 0.00 & 0.03 & 0.02 & 0.00 & 0.04 & 0.06 & 0.00 & 0.18 & 0.02 & 0.00 & 0.03 \\
link-pedigree & 1125 & 898 & 2.51 & 14 & 2 & 4 & 0.01 & 0.00 & 0.03 & 0.01 & 0.00 & 0.03 & 0.05 & 0.00 & 0.15 & 0.03 & 0.00 & 0.09 \\
linux & 213217 & 30834 & 13.83 & 9338 & 23 & 197 & 0.05 & 0.13 & 0.33 & 0.03 & 0.14 & 0.27 & 0.17 & 0.13 & 0.92 & 0.58 & 0.13 & 2.67 \\
livemocha & 2193083 & 104103 & 42.13 & 2980 & 92 & 2355 & 0.15 & 11.84 & 14.55 & 0.09 & 13.28 & 15.15 & 0.50 & 11.99 & 22.44 & 2.76 & 11.86 & 59.11 \\
loc-brightkite\_edges & 214078 & 58228 & 7.35 & 1134 & 52 & 235 & 0.03 & 0.16 & 0.34 & 0.02 & 0.17 & 0.28 & 0.11 & 0.17 & 0.82 & 0.52 & 0.17 & 2.86 \\
location & 293697 & 225486 & 2.61 & 12189 & 5 & 24 & 0.03 & 0.14 & 0.27 & 0.02 & 0.14 & 0.22 & 0.08 & 0.15 & 0.49 & 0.63 & 0.15 & 4.13 \\
mag\_geology\_coauthor & 4448428 & 2852295 & 3.12 & 1153 & 13 & 90 & 0.07 & 5.20 & 6.03 & 0.03 & 5.23 & 5.95 & 0.16 & 5.17 & 6.32 & 5.09 & 5.17 & 40.97 \\
marvel & 96662 & 19428 & 9.95 & 1625 & 18 & 135 & 0.04 & 0.05 & 0.17 & 0.03 & 0.05 & 0.14 & 0.14 & 0.05 & 0.57 & 0.34 & 0.05 & 1.22 \\
mg\_casino & 326 & 109 & 5.98 & 94 & 9 & 9 & 0.01 & 0.00 & 0.02 & 0.02 & 0.00 & 0.03 & 0.07 & 0.00 & 0.16 & 0.02 & 0.00 & 0.04 \\
mg\_forrestgump & 271 & 94 & 5.77 & 89 & 8 & 8 & 0.01 & 0.00 & 0.02 & 0.02 & 0.00 & 0.03 & 0.06 & 0.00 & 0.15 & 0.02 & 0.00 & 0.03 \\
mg\_godfatherII & 219 & 78 & 5.62 & 34 & 8 & 8 & 0.01 & 0.00 & 0.03 & 0.02 & 0.00 & 0.04 & 0.06 & 0.00 & 0.17 & 0.01 & 0.00 & 0.03 \\
mg\_watchmen & 201 & 76 & 5.29 & 33 & 7 & 7 & 0.01 & 0.00 & 0.02 & 0.02 & 0.00 & 0.03 & 0.06 & 0.00 & 0.15 & 0.01 & 0.00 & 0.03 \\
minnesota & 3303 & 2642 & 2.50 & 5 & 2 & 4 & 0.01 & 0.00 & 0.04 & 0.01 & 0.00 & 0.03 & 0.06 & 0.00 & 0.16 & 0.06 & 0.00 & 0.17 \\
moreno\_health & 10455 & 2539 & 8.24 & 27 & 7 & 34 & 0.03 & 0.00 & 0.09 & 0.02 & 0.00 & 0.07 & 0.11 & 0.00 & 0.40 & 0.11 & 0.00 & 0.41 \\
mousebrain & 16089 & 213 & 151.07 & 205 & 111 & 189 & 0.07 & 0.01 & 0.73 & 0.09 & 0.01 & 0.89 & 0.33 & 0.01 & 4.40 & 0.12 & 0.01 & 1.29 \\
movielens\_1m & 1000209 & 9746 & 205.26 & 3428 & 135 & 1394 & 0.50 & 3.52 & 15.73 & 0.32 & 3.58 & 8.97 & 1.44 & 3.57 & 32.03 & 2.19 & 3.64 & 46.71 \\
movies & 192 & 101 & 3.80 & 19 & 3 & 8 & 0.01 & 0.00 & 0.02 & 0.02 & 0.00 & 0.04 & 0.06 & 0.00 & 0.17 & 0.02 & 0.00 & 0.04 \\
muenchen-bahn & 578 & 447 & 2.59 & 13 & 2 & 4 & 0.01 & 0.00 & 0.03 & 0.01 & 0.00 & 0.03 & 0.06 & 0.00 & 0.15 & 0.02 & 0.00 & 0.06 \\
munin & 1397 & 1324 & 2.11 & 66 & 3 & 5 & 0.01 & 0.00 & 0.03 & 0.01 & 0.00 & 0.03 & 0.05 & 0.00 & 0.14 & 0.04 & 0.00 & 0.10 \\
netscience & 2742 & 1461 & 3.75 & 34 & 19 & 19 & 0.02 & 0.00 & 0.04 & 0.01 & 0.00 & 0.03 & 0.07 & 0.00 & 0.19 & 0.05 & 0.00 & 0.13 \\
offshore & 505965 & 278877 & 3.63 & 37336 & 13 & 51 & {---} & {---} & {---} & {---} & {---} & {---} & {---} & {---} & {---} & {---} & {---} & {---} \\
openflights & 15677 & 2939 & 10.67 & 242 & 28 & 82 & 0.03 & 0.01 & 0.13 & 0.03 & 0.01 & 0.10 & 0.13 & 0.01 & 0.56 & 0.16 & 0.01 & 0.74 \\
p2p-Gnutella04 & 39994 & 10876 & 7.35 & 103 & 7 & 58 & 0.03 & 0.02 & 0.14 & 0.02 & 0.02 & 0.09 & 0.11 & 0.02 & 0.44 & 0.23 & 0.02 & 1.02 \\
panama & 702437 & 556686 & 2.52 & 7015 & 62 & 64 & 0.05 & 0.58 & 0.81 & 0.02 & 0.59 & 0.76 & 0.11 & 0.57 & 1.04 & 1.58 & 0.62 & 9.47 \\
paradise & 794545 & 542102 & 2.93 & 35359 & 23 & 117 & 0.05 & 0.60 & 0.84 & 0.03 & 0.61 & 0.76 & 0.11 & 0.58 & 1.10 & 1.61 & 0.61 & 9.79 \\
photoviz\_dynamic & 610 & 376 & 3.24 & 29 & 4 & 9 & 0.01 & 0.00 & 0.03 & 0.01 & 0.00 & 0.03 & 0.06 & 0.00 & 0.16 & 0.02 & 0.00 & 0.06 \\
pigs & 592 & 492 & 2.41 & 39 & 2 & 4 & 0.01 & 0.00 & 0.02 & 0.01 & 0.00 & 0.03 & 0.05 & 0.00 & 0.15 & 0.02 & 0.00 & 0.06 \\
polblogs & 16715 & 1224 & 27.31 & 351 & 36 & 155 & 0.06 & 0.01 & 0.29 & 0.05 & 0.01 & 0.26 & 0.22 & 0.01 & 1.27 & 0.16 & 0.01 & 0.77 \\
polbooks & 441 & 105 & 8.40 & 25 & 6 & 12 & 0.02 & 0.00 & 0.04 & 0.02 & 0.00 & 0.05 & 0.07 & 0.00 & 0.25 & 0.03 & 0.00 & 0.06 \\
pollination-carlinville & 15255 & 1500 & 20.34 & 157 & 18 & 141 & 0.05 & 0.01 & 0.22 & 0.04 & 0.01 & 0.20 & 0.21 & 0.01 & 1.08 & 0.16 & 0.01 & 0.73 \\
pollination-daphni & 2933 & 797 & 7.36 & 124 & 9 & 37 & 0.02 & 0.00 & 0.06 & 0.02 & 0.00 & 0.06 & 0.09 & 0.00 & 0.31 & 0.05 & 0.00 & 0.16 \\
pollination-tenerife & 129 & 68 & 3.79 & 17 & 4 & 7 & 0.01 & 0.00 & 0.02 & 0.01 & 0.00 & 0.03 & 0.05 & 0.00 & 0.14 & 0.01 & 0.00 & 0.03 \\
pollination-uk & 16712 & 984 & 33.97 & 256 & 35 & 128 & 0.06 & 0.01 & 0.29 & 0.06 & 0.01 & 0.26 & 0.22 & 0.01 & 1.13 & 0.15 & 0.01 & 0.66 \\
ratbrain & 23030 & 503 & 91.57 & 497 & 67 & 101 & 0.07 & 0.01 & 0.19 & 0.08 & 0.01 & 0.21 & 0.28 & 0.01 & 0.83 & 0.15 & 0.01 & 0.42 \\
reactome & 147547 & 6327 & 46.64 & 855 & 62 & 234 & 0.11 & 0.10 & 0.63 & 0.08 & 0.10 & 0.42 & 0.36 & 0.10 & 1.82 & 0.53 & 0.10 & 1.99 \\
residence\_hall & 1839 & 217 & 16.95 & 56 & 11 & 44 & 0.03 & 0.00 & 0.10 & 0.04 & 0.00 & 0.13 & 0.11 & 0.00 & 0.61 & 0.05 & 0.00 & 0.18 \\
rhesusbrain & 3054 & 242 & 25.24 & 111 & 19 & 65 & 0.03 & 0.00 & 0.14 & 0.04 & 0.00 & 0.18 & 0.13 & 0.00 & 0.86 & 0.06 & 0.00 & 0.27 \\
roget-thesaurus & 3648 & 1010 & 7.22 & 28 & 6 & 26 & 0.02 & 0.00 & 0.06 & 0.02 & 0.00 & 0.06 & 0.10 & 0.00 & 0.32 & 0.06 & 0.00 & 0.18 \\
seventh-graders & 250 & 29 & 17.24 & 28 & 13 & 18 & 0.01 & 0.00 & 0.04 & 0.02 & 0.00 & 0.07 & {---} & {---} & {---} & 0.01 & 0.00 & 0.04 \\
slashdot\_threads & 117378 & 51083 & 4.60 & 2915 & 13 & 194 & 0.02 & 0.09 & 0.18 & 0.02 & 0.09 & 0.15 & 0.08 & 0.09 & 0.45 & 0.32 & 0.09 & 1.67 \\
soc-Epinions1 & 405740 & 75879 & 10.69 & 3044 & 67 & 862 & 0.04 & 0.79 & 1.19 & 0.03 & 0.80 & 1.04 & 0.15 & 0.80 & 2.45 & 0.71 & 0.81 & 6.99 \\
soc-Slashdot0811 & 469180 & 77360 & 12.13 & 2539 & 54 & 887 & 0.05 & 1.06 & 1.52 & 0.03 & 1.19 & 1.50 & 0.16 & 1.07 & 2.81 & 0.80 & 1.06 & 8.42 \\
soc-advogato & 39432 & 5167 & 15.26 & 807 & 25 & 205 & 0.05 & 0.03 & 0.28 & 0.03 & 0.03 & 0.19 & 0.18 & 0.03 & 1.06 & 0.23 & 0.03 & 1.10 \\
soc-gplus & 39194 & 23628 & 3.32 & 2761 & 12 & 61 & 0.02 & 0.01 & 0.06 & 0.01 & 0.01 & 0.05 & 0.06 & 0.02 & 0.20 & 0.16 & 0.01 & 0.55 \\
soc-hamsterster & 16630 & 2426 & 13.71 & 273 & 24 & 109 & 0.04 & 0.01 & 0.15 & 0.03 & 0.01 & 0.11 & 0.17 & 0.01 & 0.65 & 0.16 & 0.01 & 0.65 \\
soc-wiki-Vote & 2914 & 889 & 6.56 & 102 & 9 & 26 & 0.02 & 0.00 & 0.05 & 0.02 & 0.00 & 0.05 & 0.10 & 0.00 & 0.29 & 0.05 & 0.00 & 0.15 \\
\abbr{sp\_d.sc.d.\_2} & 5539 & 238 & 46.55 & 88 & 33 & 114 & 0.05 & 0.00 & 0.25 & 0.06 & 0.00 & 0.32 & 0.16 & 0.00 & 1.41 & 0.08 & 0.00 & 0.46 \\
teams & 1366466 & 935591 & 2.92 & 2671 & 9 & 258 & 0.05 & 1.18 & 1.43 & 0.03 & 1.24 & 1.41 & 0.12 & 1.20 & 1.71 & 2.00 & 1.19 & 13.77 \\
train\_bombing & 243 & 64 & 7.59 & 29 & 10 & 10 & 0.01 & 0.00 & 0.03 & 0.02 & 0.00 & 0.04 & 0.05 & 0.00 & 0.17 & 0.02 & 0.00 & 0.04 \\
tv\_tropes & 3232134 & 152093 & 42.50 & 12400 & 115 & 2916 & {---} & {---} & {---} & {---} & {---} & {---} & {---} & {---} & {---} & {---} & {---} & {---} \\
twittercrawl & 154824 & 3656 & 84.70 & 1084 & 132 & 598 & 0.18 & 0.27 & 2.06 & 0.14 & 0.28 & 1.57 & 0.57 & 0.27 & 6.75 & 0.61 & 0.27 & 5.89 \\
ukroad & 15641 & 12378 & 2.53 & 5 & 3 & 4 & 0.02 & 0.01 & 0.05 & 0.01 & 0.01 & 0.04 & 0.06 & 0.01 & 0.17 & 0.14 & 0.01 & 0.46 \\
unicode\_languages & 1255 & 868 & 2.89 & 141 & 4 & 10 & 0.01 & 0.00 & 0.03 & 0.01 & 0.00 & 0.03 & 0.06 & 0.00 & 0.15 & 0.03 & 0.00 & 0.08 \\
wafa-ceos & 93 & 26 & 7.15 & 22 & 5 & 8 & 0.01 & 0.00 & 0.02 & 0.02 & 0.00 & 0.04 & {---} & {---} & {---} & 0.01 & 0.00 & 0.02 \\
wafa-eies & 652 & 45 & 28.98 & 44 & 24 & 29 & 0.02 & 0.00 & 0.06 & 0.02 & 0.00 & 0.10 & {---} & {---} & {---} & 0.02 & 0.00 & 0.07 \\
wafa-hightech & 159 & 21 & 15.14 & 20 & 12 & 15 & 0.01 & 0.00 & 0.03 & 0.02 & 0.00 & 0.05 & {---} & {---} & {---} & 0.01 & 0.00 & 0.03 \\
wafa-padgett & 27 & 15 & 3.60 & 8 & 3 & 4 & 0.01 & 0.00 & 0.01 & 0.01 & 0.00 & 0.03 & {---} & {---} & {---} & 0.01 & 0.00 & 0.01 \\
web-EPA & 8909 & 4271 & 4.17 & 175 & 6 & 42 & 0.02 & 0.00 & 0.06 & 0.02 & 0.00 & 0.04 & 0.07 & 0.00 & 0.22 & 0.09 & 0.00 & 0.26 \\
web-california & 15969 & 6175 & 5.17 & 199 & 11 & 63 & 0.02 & 0.01 & 0.07 & 0.02 & 0.01 & 0.06 & 0.09 & 0.01 & 0.28 & 0.12 & 0.01 & 0.41 \\
web-google & 2773 & 1299 & 4.27 & 59 & 17 & 17 & 0.02 & 0.00 & 0.04 & 0.01 & 0.00 & 0.03 & 0.07 & 0.00 & 0.19 & 0.05 & 0.00 & 0.13 \\
wiki-vote & 100762 & 7115 & 28.32 & 1065 & 53 & 456 & 0.08 & 0.14 & 0.77 & 0.06 & 0.14 & 0.51 & 0.29 & 0.14 & 2.62 & 0.43 & 0.14 & 3.27 \\
wikipedia-norm & 15372 & 1881 & 16.34 & 455 & 22 & 107 & 0.04 & 0.01 & 0.17 & 0.04 & 0.01 & 0.14 & 0.17 & 0.01 & 0.74 & 0.15 & 0.01 & 0.58 \\
win95pts & 112 & 99 & 2.26 & 9 & 2 & 3 & 0.01 & 0.00 & 0.02 & 0.01 & 0.00 & 0.03 & 0.05 & 0.00 & 0.14 & 0.01 & 0.00 & 0.03 \\
windsurfers & 336 & 43 & 15.63 & 31 & 11 & 21 & 0.02 & 0.00 & 0.04 & 0.02 & 0.00 & 0.07 & {---} & {---} & {---} & 0.02 & 0.00 & 0.05 \\
word\_adjacencies & 425 & 112 & 7.59 & 49 & 6 & 17 & 0.01 & 0.00 & 0.03 & 0.02 & 0.00 & 0.05 & 0.07 & 0.00 & 0.25 & 0.02 & 0.00 & 0.06 \\
zewail & 54182 & 6651 & 16.29 & 331 & 18 & 142 & 0.06 & 0.04 & 0.27 & 0.04 & 0.04 & 0.19 & 0.19 & 0.04 & 1.04 & 0.31 & 0.04 & 1.34 \\
\end{longtable}

%% file: biblio.bib
@InProceedings{drange2023computing,
  author    =  {Drange, Pål Grønås and Greaves, Patrick and Muzi, Irene and Reidl, Felix},
  title     =  {{Computing Complexity Measures of Degenerate Graphs}},
  booktitle =  {18th International Symposium on Parameterized and Exact Computation (IPEC 2023)},
  pages     =  {14:1--14:21},
  series    =  {Leibniz International Proceedings in Informatics (LIPIcs)},
  ISBN      =  {978-3-95977-305-8},
  ISSN      =  {1868-8969},
  year      =  2023,
  volume    =  285,
  publisher =  {Schloss Dagstuhl -- Leibniz-Zentrum f{\"u}r Informatik},
  address   =  {Dagstuhl, Germany},
  doi       =  {10.4230/LIPIcs.IPEC.2023.14},
}

@article{matulaDegeneracy1983,
	title = {Smallest-last ordering and clustering and graph coloring algorithms},
	volume = {30},
	issn = {0004-5411},
	url = {https://dl.acm.org/doi/10.1145/2402.322385},
	doi = {10.1145/2402.322385},
	number = {3},
	urldate = {2023-06-30},
	journal = {Journal of the ACM},
	author = {Matula, David W. and Beck, Leland L.},
	year = {1983},
	pages = {417--427},
}

@article{chen_strong_2006,
	title = {Strong computational lower bounds via parameterized complexity},
	volume = {72},
	issn = {00220000},
	url = {https://linkinghub.elsevier.com/retrieve/pii/S0022000006000675},
	doi = {10.1016/j.jcss.2006.04.007},
	language = {en},
	number = {8},
	urldate = {2023-01-11},
	journal = {Journal of Computer and System Sciences},
	author = {Chen, Jianer and Huang, Xiuzhen and Kanj, Iyad A. and Xia, Ge},
	year = {2006},
	pages = {1346--1367},
}

@article{hardnessGeneralizedColoring2023,
  title = {Hardness of the Generalized Coloring Numbers},
  author = {{Breen-McKay}, Michael and Lavallee, Brian and Sullivan, Blair D.},
  year = {2023},
  month = mar,
  journal = {European Journal of Combinatorics},
  pages = {103709},
  issn = {0195-6698},
  doi = {10.1016/j.ejc.2023.103709},
  urldate = {2023-09-12},
}

@inproceedings{abboudRadius2016,
  title = {Approximation and Fixed Parameter Subquadratic Algorithms for Radius and Diameter in Sparse Graphs},
  booktitle = {Proceedings of the Twenty-Seventh Annual {{ACM-SIAM}} Symposium on {{Discrete}} Algorithms},
  author = {Abboud, Amir and Williams, Virginia Vassilevska and Wang, Joshua},
  year = {2016},
  month = jan,
  series = {{{SODA}} '16},
  pages = {377--391},
  publisher = {{Society for Industrial and Applied Mathematics}},
  address = {USA},
  urldate = {2023-09-13},
  isbn = {978-1-61197-433-1}
}

@inproceedings{twoAdmissibility25,
  author       = {Christine Awofeso and
                  Patrick Greaves and
                  Oded Lachish and
                  Felix Reidl},
  title        = {A Practical Algorithm for 2-Admissibility},
  booktitle    = {23rd International Symposium on Experimental Algorithms, {SEA} 2025,
                  July 22-24, 2025, Venice, Italy},
  series       = {LIPIcs},
  volume       = {338},
  pages        = {3:1--3:19},
  publisher    = {Schloss Dagstuhl - Leibniz-Zentrum f{\"{u}}r Informatik},
  year         = {2025},
  url          = {https://doi.org/10.4230/LIPIcs.SEA.2025.3},
  doi          = {10.4230/LIPICS.SEA.2025.3},
}

@article{neighbourhoodComplexity2019,
  author       = {Felix Reidl and
                  Fernando S{\'{a}}nchez Villaamil and
                  Konstantinos S. Stavropoulos},
  title        = {Characterising bounded expansion by neighbourhood complexity},
  journal      = {Eur. J. Comb.},
  volume       = {75},
  pages        = {152--168},
  year         = {2019},
  url          = {https://doi.org/10.1016/j.ejc.2018.08.001},
  doi          = {10.1016/J.EJC.2018.08.001},
  bibsource    = {dblp computer science bibliography, https://dblp.org}
}

@inproceedings{kernelDomsetBndExp2016,
  author       = {P{\aa}l Gr{\o}n{\aa}s Drange and
                  Markus Sortland Dregi and
                  Fedor V. Fomin and
                  Stephan Kreutzer and
                  Daniel Lokshtanov and
                  Marcin Pilipczuk and
                  Michal Pilipczuk and
                  Felix Reidl and
                  Fernando S{\'{a}}nchez Villaamil and
                  Saket Saurabh and
                  Sebastian Siebertz and
                  Somnath Sikdar},
  title        = {Kernelization and Sparseness: the Case of Dominating Set},
  booktitle    = {33rd Symposium on Theoretical Aspects of Computer Science, {STACS}
                  2016, February 17-20, 2016, Orl{\'{e}}ans, France},
  series       = {LIPIcs},
  volume       = {47},
  pages        = {31:1--31:14},
  publisher    = {Schloss Dagstuhl - Leibniz-Zentrum f{\"{u}}r Informatik},
  year         = {2016},
  url          = {https://doi.org/10.4230/LIPIcs.STACS.2016.31},
  doi          = {10.4230/LIPICS.STACS.2016.31},
  bibsource    = {dblp computer science bibliography, https://dblp.org}
}

@inproceedings{kernelStructuralBndExp13,
  author       = {Jakub Gajarsk{\'{y}} and
                  Petr Hlinen{\'{y}} and
                  Jan Obdrz{\'{a}}lek and
                  Sebastian Ordyniak and
                  Felix Reidl and
                  Peter Rossmanith and
                  Fernando S{\'{a}}nchez Villaamil and
                  Somnath Sikdar},
  title        = {Kernelization Using Structural Parameters on Sparse Graph Classes},
  booktitle    = {Algorithms - {ESA} 2013 - 21st Annual European Symposium, Sophia Antipolis,
                  France, September 2-4, 2013. Proceedings},
  series       = {Lecture Notes in Computer Science},
  volume       = {8125},
  pages        = {529--540},
  publisher    = {Springer},
  year         = {2013},
  url          = {https://doi.org/10.1007/978-3-642-40450-4\_45},
  doi          = {10.1007/978-3-642-40450-4\_45},
  bibsource    = {dblp computer science bibliography, https://dblp.org}
}

@article{dvorakDomset2013,
  author    = {\Zdenek \Dvorak},
  title     = {Constant-factor approximation of the domination number in sparse graphs},
  journal   = {Eur. J. Comb.},
  volume    = 34,
  number    = 5,
  pages     = {833--840},
  year      = 2013,
  url       = {https://doi.org/10.1016/j.ejc.2012.12.004},
  doi       = {10.1016/j.ejc.2012.12.004},
  bibsource = {dblp computer science bibliography, https://dblp.org}
}

@book{Sparsity,
  author    = {Jaroslav \Nesetril and
             Patrice {Ossona de Mendez}},
  Publisher = {Springer},
  Series = {Algorithms and Combinatorics},
  Title = {{Sparsity: Graphs, Structures, and Algorithms}},
  Volume = {28},
  Year = {2012}
}

@article{zhuColouring2009,
  title = {Colouring Graphs with Bounded Generalized Colouring Number},
  author = {Zhu, Xuding},
  year = {2009},
  month = sep,
  journal = {Discrete Mathematics},
  series = {Combinatorics 2006, {{A Meeting}} in {{Celebration}} of {{Pavol Hell}}'s 60th {{Birthday}} ({{May}} 1--5, 2006)},
  volume = {309},
  number = {18},
  pages = {5562--5568},
  issn = {0012-365X},
  doi = {10.1016/j.disc.2008.03.024},
  urldate = {2025-09-16},
}
